\documentclass[11pt,b5paper,twoside, headrule]{amsart}

\usepackage{amsfonts, amsmath, amssymb,latexsym}
\usepackage{epsfig}
\usepackage[curve]{xy}
\usepackage{algorithmic}
\usepackage{algorithm}
\usepackage{enumerate}
\usepackage{framed}
\usepackage{hyperref}
\usepackage{mathtools}

\usepackage{chngcntr}
\footskip=10mm\parskip 0.2cm\addtocounter{page}{0}
\newtheorem{thm}{Theorem}[section]
\newtheorem{cor}[thm]{Corollary}
\newtheorem{lem}[thm]{Lemma}
\newtheorem{prop}[thm]{Proposition}
\newtheorem{conj}[thm]{Conjecture}
\newtheorem{alg}[thm]{Algorithm}

\theoremstyle{definition}
\newtheorem{defn}[thm]{Definition}
\newtheorem{examples}[thm]{Examples}
\newtheorem{example}[thm]{Example}

\newtheorem{construction}[thm]{Construction}
\newtheorem{prob}[thm]{Problem}
\newtheorem{question}[thm]{Question}
\theoremstyle{remark}
\newtheorem{rem}[thm]{Remark}
\numberwithin{equation}{section}

\title[Ring learning with errors]{Ring learning with errors: a crossroads between postquantum cryptography, machine learning and number theory}

\author{\sc Iv\'an Blanco-Chac\'on}
\address{Department of Mathematics, School of Science\\
University of Alcal\'a de Henares\\
Ctra. Madrid-Barcelona Km. 33,600\\
Alcal\'a de Henares, Spain}
\email{ivan.blancoc@uah.es}
\thanks{Partially supported by Science Foundation Ireland 13/IA/1914 and MTM2016-79400-P.}


\begin{document}
%=========================== Baselineskip ===============================
\renewcommand\baselinestretch{1.2}
\renewcommand{\arraystretch}{1}
\def\base{\baselineskip}
%========================= Begin Document--Fonts--and other ============================
\font\tenhtxt=eufm10 scaled \magstep0 \font\tenBbb=msbm10 scaled
\magstep0 \font\tenrm=cmr10 scaled \magstep0 \font\tenbf=cmb10
scaled \magstep0

%%================================head--document=======================

\def\evenhead{{\protect\centerline{\textsl{\large{I. Blanco}}}\hfill}}

\def\oddhead{{\protect\centerline{\textsl{\large{On the non vanishing of the cyclotomic $p$-adic $L$-functions}}}\hfill}}

\pagestyle{myheadings} \markboth{\evenhead}{\oddhead}

\thispagestyle{empty}

\maketitle

\begin{abstract}The present survey inteds to serve as a comprehensive account of the main areas of the cryptography based on the Ring Learning With Errors Problem. We cover the major topics, from their mathematical foundations to the main primitives, as well as several open ends and recent progress with an emphasis in the connections with algebraic number theory. This work is based to a certain extent on an invited course and a seminar given by the author at the Basque Center for Applied Mathematics in 2018 and at the ICIAM 2019. Our aim is to provide an introduction to the topic for graduate students with a background in algebra or number theory.
\end{abstract}

\bigskip
\section{Introduction}

According to MIRACL Labs, it is estimated that a quantum computer capable of breaking most of modern cryptography will be built in the next 10-15 years (20-25 years according to estimates made public in the last NIST call for the standardisation of postquantum primitives). All of cryptography is built on supposedly \emph{hard}\footnote[1]{In a sense which will be made clear in Section 2.} mathematical problems, most of which,  like integer factorisation or the discrete logarithm problem, become relatively easy in the context of a working quantum computer. In response to this threat there is a need to migrate from these vulnerable constructs to constructs known to remain strong even in a post-quantum world. 

An example of such a hard problem is the shortest vector problem in general lattices, which is known to be NP-hard (at least for a very small approximation factor). While there already exist post-quantum solutions for much of standard cryptography, like public key encryption and digital signature, it is currently unclear how some of the more elaborate protocols, like those seeking for integrity or non-repudiation can be successfully migrated. In particular in the last 10+ years bilinear pairings on elliptic curves have opened up many new possibilities, which might likely be rendered insecure in a postquantum world. Already commercial products based on bilinear pairings have found applications in the `real world', and so much work must be done to ensure that we will be able to retain this functionality into the future.

At the same time there is much fundamental work to be done on the postquantum primitives themselves. A major decision is to choose between one or various of the following  technologies, for each security/integrity demand:

\begin{itemize}
\item[a)]Code based cryptography (\cite{oberbeck}) is built on the infeasibility of syndrome decoding for general linear error-correcting codes over finite fields.

\item[]

\item[b)]  Multivariate based cryptography (\cite{ding}) is based on the fact that solving general systems of multivariate polynomial equations over finite fields is proved to be NP-hard.

\item[]

\item[c)] Supersingular isogeny based cryptography (\cite{leo}), is a protocol for key exchange, analogous to Diffie-Hellman, but the cyclic groups present here are attached to supersingular elliptic curves defined over finite fields.

\item[]

\item[d)] Finally, lattice based cryptography, admits a large number of different formulations and constructions. This report focuses on one of the most promising lattice-based technologies: Ring Learning With Errors (RLWE). This scheme is based on the RLWE problem, which is based in turn on the difficulty of solving the shortest vector problem (SVP) on ideal lattices. 
\end{itemize}

At the time of writing, code, lattice and multivariate-based methods seem to be the strongest contenders, as they appear to have the flexible structure needed on which to base more complex protocols. Within these three categories, the lattice-based one has by far a larger number of non-broken primitives/protocols.  

Lattice-based cryptography has a relatively mature history, primarily due to the work done by the early proponents of the related NTRU cryptosystem (\cite{ntru}). This was a patented technology which enjoyed some minor success, but never really gained traction, as when it was invented, a quantum computer still seemed very far off. Its patents have now expired.

RLWE first came to prominence with the paper by Lyubashevsky, Peikert and Regev (\cite{LPR}). A key-exchange algorithm proposed by them has been recently optimized and implemented by Alkim, Ducas, P\"oppelmann, and Schwabe (\cite{ADPS}). This has been implemented by Google in a well-publicised experiment (\cite{google}). In recent times there have been many implementation improvements, see for example the recent paper by Scott (\cite{scott}). So there can be no doubting the practicality of the technology, opinions supporting this view include those of a good number of researchers in Intel Labs and MIRACL Labs.

RLWE is built on an earlier scheme: the Learning with Errors (LWE) problem, which admits a security reduction from the SVP on arbitrary lattices, but with a much larger approximation factor than the one for which SVP is proved to be NP-hard. Of course, this is not a formal  hardness guarantee for LWE but it can be regarded as a clue of its strength. Moreover, no polynomial-time attack has been found against LWE yet. 

The main disadvantage of LWE is a quadratic overhead in the key sizes, which is overcome in the RLWE scenario, at the cost of being backed in the SVP over just ideal lattices, which even if based on experience is widely believed to be intractable, there is no formal proof at the moment, and for no appromation factor. 

In spite of that, the RLWE variant appears to be eminently practical: like most post-quantum proposals, RLWE key sizes are much larger than those of non post-quantum methods, but the required computing power is usually much smaller. For example while an elliptic curve based cryptosystem might use keys of 256 bits, an equivalent system based on RLWE might require keys of 4096 bits to grant the same security level, while running maybe 10-100 times faster (\cite{leclerq}). These differences might be seen as balancing each other out. Furthermore, 30\% of the surviving proposals for the NIST are based on RLWE.

Our report is structured as follows: 

In section 2 we provide a quick introduction to the different features of cryptography and introduce the main terms and facts on complexity as they show up in the literature. We provide several examples, elaborating on those presented in the course by the author. 

In section 3 we expose the main concepts of lattice-based cryptography. We focus on the classical LWE, over which RLWE is built and discuss its advantadges and drawbacks, as well as different attacks against weak instantiations, which will be exploited in the RLWE scenario in Section 6.

Section 4 is a quick overview of a few key concepts in algebraic number theory: rings of integers, canonical embedding, and other topics. These pieces make the foundations of RLWE, but the reader who is familiar with this material can safely skip it. 

Section 5 introduces the RLWE problem in its various formulations. In particular, we carefully discuss the Polynomial Learning With Errors problem (PLWE), which appeared in the literature before RLWE (\cite{stehle2}). We discuss the equivalence between both problems and explain some recent advances in this topic: in particular we comment on recent work by the author (\cite{blanco}) which gives a partial answer in the cyclotomic case, the most interesting from a cryptographic point of view. Besides, we explain the hardness result which backs RLWE and describe in full detail the LPR crypstosystem, as presented in \cite{LPR}. We close the chapter by presenting a key exchange protocol based in RLWE (\cite{dingkex}).

Section 6  is a summary of several attacks against the RLWE cryptosystem. They reduce to LWE or to PLWE attacks and allow to discard insecure choices of parameters. The search for secure instantiations motivates some number theoretical problems and conjectures which we also discuss.

Section 7 is for RLWE-based digital signatures and homomorphic encryption, a functionality which is gaining much interest nowadays, since it allows to solve a good number of logistic and security problems in cloud computing and storing. We close the survey by discussing in detail some (second round) NIST figures. 

A couple of remarks to end this introduction: first, by a \emph{polynomial time algorithm} we mean an algorithm for which there exists a polynomial $p(x)\in\mathbb{R}[x]$ and a \emph{size} function on the family of the algorithm inputs $x(n)$, such that the time it takes to run the algorithm on input $n$ is $p(x(n))$. Second, we will use sometimes the $\tilde{O}$-notation: a function $f(x)$ is $\tilde{O}(g(x))$ if it is $O(g(x)log^k(x))$ for some $k$.

\textbf{Acknowledgements:} The author thanks Gary McGuire for carefully reading a preliminary version of this survey, to Mike Scott for providing most of the practical highlights on RLWE and to the Basque Center for Applied Mathematics for their invitation to give this course and to take part in the postquantum cryptography mini-symposum at ICIAM 2019. Active and insightful discussion with the audience of the course and seminar, and in particular with Sebasti\'a Xamb\'o set the author to write this work.

\section{Post-quantum cryptography}

\subsection{Cryptography features}
Requirements such as confidentiality and proofs of identity are crucial in electronic financial and legal transactions, while some other features like non-repudiation or operating on encrypted data (homomorphic encryption) are gaining much traction within the last few years. We examine here most of these functionalities.

The best known cryptographic problem is confidentiality. This is attained by the use of well-designed encryption/decription schemes. 

To start with, we fix a finite alphabet $\mathcal{A}$, with some mathematical structure such as an abelian group or a field (e.g. the finite field $\mathbb{F}_q$ for $q=p^t$ and $p$ prime, or an elliptic curve over this field). We consider three sets $\mathcal{K}\subset \mathcal{A}^n$ (keys), $\mathcal{M}\subset \mathcal{A}^N$ (plaintexts) and $\mathcal {C}\subset \mathcal{A}^N$ (ciphertexts) with $n\ll N$. Finally, we consider a set $\Lambda\subset\mathbb{N}$ which parametrizes the level of security, i.e., the larger $\lambda\in\Lambda$, the safer the scheme. 

\begin{defn}[Cipher schemes] A cipher over $(\mathcal{K},\mathcal{M},\mathcal{C})$ is a family of pairs of efficiently computable functions $\{(E_{\lambda},D_{\lambda})\}_{\lambda\in\Lambda}$ where for each $\lambda$, $E_{\lambda}: \mathcal{K}\times \mathcal{M}\to \mathcal{C}$ (encryption function) and $D_{\lambda}: \mathcal{K} \times \mathcal{C}\to \mathcal{M}$ (decryption function) are such that for each key $k\in \mathcal{K}$ and for each plaintext $m\in \mathcal{M}$, the following correctness property holds:
$$
D_{\lambda}(k,E_{\lambda}(k,m))=m.
$$
That is to say, decryption undoes encryption.
\end{defn}
Efficiently computable means that both $E_{\lambda}$ and $D_{\lambda}$ can be computed by an algorithm which is polynomial in the security parameter $\lambda$, i.e., there exist polynomials $p(x), q(x)\in\mathbb{R}[x]$ only depending on the scheme, such that for each $\lambda\in\Lambda$, $k\in\mathcal{K}$, $m\in \mathcal{M}$, and $c\in \mathcal{C}$, the number of steps to compute $E_{\lambda}(k,m)$ (resp. $D_{\lambda}(k,c)$) is upper bounded by $p(\lambda)$ (resp. $q(\lambda)$). Moreover, the algorithm for $E_{\lambda}$ can be probabilistic, while $D_{\lambda}$ should always be deterministic. 

Since in our definition both the encryption and decryption parties have the same key (i.e., the scheme is symmetric), they should agree beforehand on that key somehow. For instance, they might do it physically in a secret meeting but they can also use a digital key exchange protocol. As usual, any arbitrary legitimate sender (receiver) will be called Alice (Bob), and any arbitrary eavesdropper will be called Eve.

\begin{defn}[Key exchange protocol] A key exchange protocol is an efficient method for Alice and Bob to agree on a key through a (potentially non-safe) channel. One of the most famous protocols is Diffie-Hellman's (DH)\footnote{There are more general versions of the Diffie-Hellman problem, not known to be equivalent to a discrete logarithm problem, but here we stick to its version over finite fields, which by construction is so.}, where Alice and Bob start by agreeing on a finite feld $\mathbb{F}_q$ and a primitive root $g$, namely, a generator of the cyclic multiplicative group $\mathbb{F}_q^*$. The pair $(q,g)$ is made public, and to agree on a private key,  Alice selects an integer $a$ and Bob selects an integer $b$. Then, Alice sends $g^a$ modulo $q$ to Bob, who on receiving it, raises it to $b$ modulo $q$, getting $g^{ab}$ modulo $q$. Next, Bob sends $g^b$ modulo $q$ to Alice, who raises it to $a$, obtaining also $g^{ab}$ modulo $q$, the agreed private key.
\end{defn}

Notice that without knowledge of $a$ or $b$, Eve cannot obtain $g^{ab}$ from $g^a$ and $g^b$ in an efficient manner (on a classic computer!), the main obstruction being the unfeasibility of the discrete logarithm, namely, to obtain $a$ from $g^a$ modulo $q$, if $g$ is known. Nowadays, a combined usage of Diffie-Hellman (or some variant) with a suitable symmetric cipher is used in most internet protocols, like TLS or TCP/IP. A variant of DH is ECDH, where the multiplicative group $\mathbb{F}_q^{*}$ is replaced by the additive group of an elliptic curve over $\mathbb{F}_q$.

\begin{defn}[Digital signatures] A signature scheme is a pair $(G,D)$, where $G:\Lambda\to \mathcal{K}$ is an efficient key generating probabilistic algorithm, and $D=\{ (S_{\lambda},V_{\lambda}) \}_{\lambda\in\Lambda}$ is a family of pairs of efficiently computable\footnote{by polynomial-time probabilistic algorithms.} functions $S_{\lambda}:\mathcal{K}\times \mathcal{M}\to \mathcal{T}\subseteq \mathcal{A}^r $ (space of \emph{tags}) and $V_{\lambda}: \mathcal{K}\times \mathcal{M} \times \mathcal{T}\to \{0,1\}$ such that whenever $k_s$ (secret key) and $k_p$ (public key) are sampled from $G$ on security level $\lambda$, then, for every message $m\in \mathcal{M}$:
$$
pr\left[ V_{\lambda}(k_p,m,S_{\lambda}(k_s,m))=1\right]=1,
$$
\end{defn}
For a security level $\lambda$, $S_{\lambda}$ is called the signature function and $V_{ \lambda}$ the verification function, which returns 1 if the signature is valid and 0 otherwise, and the correctness of the scheme means that on a message $m$ and a secret key $k_s$, the signature function produces a tag $S_{\lambda}(k_s,m)$, which is verified as valid by the verifying function with probability $1$, given the message $m$ and the public key $k_p$. This scheme provides a proof that the message was signed by a known signatory (authentication) and the signatory cannot deny having signed the message (non-repudiation). Classic designs of digital signature schemes include Rabin's algorithm, Lamport schemes and Merkle trees, as well as RSA-based protocols (\cite{cryptostanford} 13.3.1).

Integrated Encryption Schemes (signcryption schemes) implement both encryption and authentication. Two of the most commonly used are ECIES, which operates with elliptic curves and  DLIES, which operates over $\mathbb{F}_q$.

\begin{defn}[Homomorphic encryption\footnote{From now on, to ease notation, we will omit the $\lambda$-subscripts unless it results in ambiguity.} ] Let $(E,D)$ be a cipher over $(\mathcal{K},\mathcal{M},\mathcal{C})$ where $\mathcal{M}$ and $\mathcal{C}$ are abelian groups under the operations $*_{\mathcal{M}}$ and $*_{\mathcal{C}}$ respectively. The cipher is said to be homomorphic if for each key $k\in \mathcal{K}$ and plaintexts $m_1,m_2\in \mathcal{M}$, it is
$$
E(k,m_1)*_{\mathcal{C}}E(k,m_2)=E(k,m_1*_{\mathcal{M}}m_2).
$$
\end{defn}

\begin{example}RSA encryption is homomorphic. Indeed, for an RSA integer $N=pq>1$ and an exponent $e$ modulo $\phi(N)$ with inverse $d$, encryption goes as $x\mapsto x^e\pmod{N}$, which clearly commutes with the product modulo $N$, but not with the sum.
\end{example}

When in addition, $\mathcal{M}$ and $\mathcal{C}$ have ring structure and encryption commutes with both ring operations, the cipher is said to be fully homomorphic (FHE). Notice that RSA is not fully homomorphic.

Homomorphic encryption allows to perform operations on the plaintext by operating directly on the ciphertexts, i.e., without decrypting first. This is relevant when the operations are outsourced and performed over a non-trustable server. Applications of homomorphic encryption include encrypted database queries, cloud computing, genetic computing, health data management or outsourced generation of blockchain addresses.

\subsection{P, NP, NP-hard and NP-complete} The author has often seen that the terms \emph{intractable}, \emph{unfeasible}, and \emph{hard}, are used in the postquantum cryptography literature in a rather loose (at best!) manner and this may lead to believe that certain computational problems enjoy certain complexity guarantees that they simply have not. We make here precise the main terms that usually appear in the problems which back lattice cryptography.

\begin{defn}[The P and NP classes] The P class consists of the decission problems whose solution can be found on a deterministic Turing machine in polynomial time in the input size. The NP class consists of  the decission problems for which a putative solution can be checked to be a real solution or not in polynomial time on a deterministic Turing machine on the input size. 
\label{defnp}
\end{defn}

Equivalently, the NP-class consists of the decission problems such that a solution can be found in polynomial time on a non-deterministic Turing machine:  indeed, assuming Definition \ref{defnp} for the NP class, an algorithm based on a non-deterministic Turing Machine can be built in two steps; the first is a non-determininstic guess about the solution, and the second consists of a polynomial deterministic algorithm that verifies if the guess is a solution(cf. \cite{alsu} pag. 283 for details). A common misconception is that the \emph{NP} term stands for \emph{non-polynomial} when in fact it stands for \emph{non-deterministic polynomial acceptable problems}.

A note of caution: as we pointed out at the end of the introduction, the term \emph{in polynomial time} means that the time it takes to solve a problem is, on input $n$, upper bounded by a polynomial in $x(n)$ where $x$ is a \emph{size} function. The most used size function is the logarithm, as we can regard it, essentially as the number of digits, a \emph{true size} of the input. Hence, a brute force attack on DLP for $\mathbb{F}_p$ takes $p-1$ powers and checks, which is polynomial in $p$ but exponential in $\log(p)$. There are classical (non-quantum) algorithms which drastically reduce the order, like the number field sieve (subexponential), but none of them is polynomial in $\log(p)$. We refer the reader to Chapter 4 of this work for a summary on number fields and their key properties and to \cite{nfs} for an exposition of the number field sieve method.

\begin{example}The problem of primality testing, i.e. deciding whether a positive integer is prime or not  is NP: indeed, given a natural number $n>1$ and $b\leq n$, the Euclidean algorithm can be used to check if $b\mid n$ in approximalety $\log(b)$ operations. Moreover, in a major breakthrough, Agrawal, Kayal and Saxena proved that primality testing is also a P problem. 
\end{example}

\begin{example}The problem of factoring, namely to return a proper factorisation $n=pq$ with $1<p,q<n$ of an input $n\in\mathbb{N}$ is also NP: a pair $(p,q)\in\mathbb{N}^2$ can be checked to be (or not) a non trivial factorisation of $n$ by performing approximately $\log(q)^2$ multiplications, if $q\geq p$. 
\end{example}

Two celebrated algorithms due to Peter Shor solve the factoring problem and the DLP in polynomial time on a quantum computer (\cite{shor}). To factor a positive integer $n$, Shor's algorithm runs over all the integers in the range $\{1,..,n\}$. For $1 < a < n$, if $a$ is a unit modulo $n$, the algorithm calls a sub-routine to compute the order of $a$ modulo $n$. With this period, the algorithm produces a non-trivial factor of $n$ with arbitrarily large probability in polynomial time. The order-finding sub-routine is run on a quantum computer, but the use of the order to produce a factor is classical.

In fairness, this does not mean that the problem of factoring is in the P-class, as a (probabilistic) quantum algorithm is not equivalent, in general, to a Turing or sequential machine. 

\begin{defn}[Reduction] We say that a problem A admits a reduction to a problem B if any instance of A can be transformed to an instance of B in polynomial time, namely, if solving B suffices for solving A with the same order of complexity.\footnote{By \emph{order of complexity} we mean \emph{polynomial, superpolynomial, subexponential} and \emph{exponential}. We stick to these orders as they are enough for our analysis.}
\end{defn}

Informally, NP-hard and NP-complete problems are those at least as hard as those in the NP-class, but while NP-complete problems belong to NP, NP-hard ones need not to. More precisely:

\begin{defn} The NP-hard class consists of those problems A such that every problem in NP can be reduced to A in polynomial time. The NP-complete class consists of those NP problems which are NP-hard.
\end{defn}

\begin{example} The prime factorisation problem, i.e. to return all the prime factors with multiplicity of an input $n\geq 1$, is clearly NP: checking if a putative solution is a prime factorisation of $n$ can be done in (deterministic) polynomial time. However it is not known if the prime factorisation is NP-hard (and hence NP-complete). It is expected, moreover, not to be in the P class.
\end{example}

%\begin{example}There is a polynomial time quantum reduction from the discrete logarithm problem (DLP) to the the order-finding problem.
 
%Indeed, given $A=g^a\pmod {N}$, use the polynomial time quantum sub-routine in Shor's algorithm to find $r$, the order of $A$. Second, use Shor's algorithm to find $\phi(N)$. To compute $\phi(N)$, Shor's algorithm has to be applied at most $\omega(N)$ times, where $\omega(N)$ is the number of prime divisors of $N$, whose asymptotic growth is $O(\log(N)/\log \log(N))$.

%Now, observe that $r=\frac{\phi(N)}{gcd(a, \phi(N))}$. Thus, $gcd(\frac{ar}{\phi(N)},r)=1$. Hence, it suffices to test, for $k\geq 2$, if $\phi(N)\mid kr$ (by the Euclidean algorithm), and in that case, also by the Euclidean algorithm, if $\frac{kr}{\phi(N)}$ is coprime to $r$. If so, check if $g^k=A$.
%\end{example}

So, a quantum computer would render insecure both RSA and Diffie-Hellman. Even more, Tate and Weil's pairings allow to reduce ECDLP to DLP (\cite{menezes}), a reduction which is even polynomial (although probabilistic) on supersingular curves, hence, the elliptic version on Diffie-Hellman should also be avoided in a post-quantum scenario. This is a reason to consider schemes which use pairing-free abelian varieties, hence other than elliptic curves. Jacobians of hyperelliptic curves are known to be good candidates but beyond genus 3, the complexity of finding explicit equations and explicit computations for the addition law render them unfeasible.

Finally, another well-known problem is whether $P\neq NP$ or not. If equality held, all cryptographic (classic and postquantum) primitives based on NP problems would be useless. On the contrary, if, as it is widely believed, $P\neq NP$, then every NP-hard problem would be non-polynomial, hence suitable for cryptography: indeed, if $\Lambda$ is NP-hard, in case $P\neq NP$, take $B$ in $NP\setminus P$. Then $\Lambda$ cannot be polynomial (otherwise, $B$ would be so).

But for the moment, lacking a proof of $P\neq NP$, all we can say is that NP-hard problems are \emph{strongly} expected to be suitable for (postquantum) cryptography.

\section{Lattice based cryptography}The security of lattice-based schemes relies on two problems which are expected to be intractable on a quantum computer, as we explain next. By length, we mean Euclidean length, denoted $||\phantom{a}||$. 

\begin{defn}A lattice in $\mathbb{R}^n$ is a pair $(\Lambda,\rho)$ where $\Lambda$ is a finitely generated and free subgroup of the additive group $(\mathbb{R}^n,+)$ and $\rho:\Lambda \to \mathbb{Z}^n$ is an isomorphism. We denote by $\lambda_1(\Lambda)$ the minimal length among the set of non-zero elements of $\Lambda$.
\end{defn}
Notice that our definition has implicit the feature of being of full rank. There are more general definitions but this will be enough for us. %Likewise, by the very definition, every lattice has a (non-unique) $\mathbb{Z}$-basis: a minimal set of free generators as $\mathbb{Z}$-module, which can be identified with its image in $\mathbb{Z}^n$.

\begin{example}In the ring of Gaussian integers $\mathbb{Z}[i]=\left\{a+bi; a,b\in\mathbb{Z}\right\}$, identifying $(\mathbb{C},+)$ with $(\mathbb{R}^2,+)$, we can impose a lattice structure in (at least) two ways:
\begin{equation}
\begin{array}{ccc}
\rho_1: \mathbb{Z}[i] & \to & \mathbb{Z}^2\\
a+bi & \mapsto & (a,b)
\end{array}
\end{equation}
or 
\begin{equation}
\begin{array}{ccc}
\rho_2:\mathbb{Z}[i] & \to & \mathbb{Z}^2\\
a+bi & \mapsto & (a+b,a-b).
\end{array}
\end{equation}
%Both $\rho_1$ and $\rho_2$ make $\mathbb{Z}[i]$ a full rank lattice in $\mathbb{R}^2$.
\end{example}

%\begin{example}The ring $\mathbb{Z}[\sqrt{2}]=\left\{a+b\sqrt{2};\mbox{ }a,b\in\mathbb{Z}\right\}$ can be made into a (non full rank) lattice in $\mathbb{R}$ via $\rho_1=Id$, the identity, or into a full rank lattice in $\mathbb{R}^2$ via
%\begin{equation}
%\begin{array}{ccc}
%\rho_2:\mathbb{Z}[\sqrt{2}] & \to & \mathbb{R}^2\\
%a+b\sqrt{2} & \mapsto & (a+b\sqrt{2},a-b\sqrt{2}).
%\end{array}
%\end{equation}
%
%\end{example}

\begin{defn}Let $(\Lambda,\rho)$ be a lattice in $\mathbb{R}^n$ with basis $\mathcal{B}=\left\{v_1,...,v_n\right\}$. The fundamental parallelogram of $\Lambda$ associated to $\mathcal{B}$ is:
$$
\mathcal{F}(\mathcal{B})=\left\{\sum_{i=1}^n\lambda_iv_i:\mbox{ with }0\leq \lambda <1\right\}.
$$
\end{defn}

%\begin{figure}
%  \centering
%   \includegraphics[width=0.7\textwidth]{fdtal.png}
%\caption{Fundamental parallelogram, in darker blue. Source: Wikipedia (by \'Alvaro Lozano Robledo)}
%\end{figure}

\begin{prob}[SVP] The shortest vector problem (SVP) is, on input of an arbitrary lattice $\Lambda$ in $\mathbb{R}^n$, together with a basis, to determine a vector $x\in\Lambda$ with length $\lambda_1(\Lambda)$. For $\gamma>0$, the $\gamma$-approximate shortest vector problem ($\gamma$-SVP) is to determine a non-zero vector $x\in\Lambda$ with $||x||\leq \gamma\lambda_1(\Lambda)$.
\end{prob}

%\begin{prob}[The approximate Shortest Vector Problem] For $\gamma>0$, the $\gamma$-approximate shortest vector problem ($\gamma$-SVP) is, on input of an arbitrary lattice $\Lambda$ in $\mathbb{R}^n$, together with a $\mathbb{Z}$-basis, to determine a non-zero vector $x\in\Lambda$ with length smaller than $\gamma\lambda_1(\Lambda)$.
%\end{prob}

%\begin{figure}
%  \centering
%   \includegraphics[width=0.7\textwidth]{SVP.png}
%\caption{Illustration of the shortest vector problem (basis vectors in blue, shortest vector in red). Source: Wikipedia (by Sebastian Schmittner)}
%\end{figure}

\begin{prob}[CVP]The closest vector problem (CVP) is, on input of an arbitrary lattice $\Lambda$ in $\mathbb{R}^n$, together with a basis and a point $y\in\mathbb{R}^n$, to find $x_y\in\Lambda$ such that
$$
||y-x_y||=\min_{x\in\Lambda}||x-y||.
$$
\end{prob}

%\begin{figure}
%  \centering
%   \includegraphics[width=0.7\textwidth]{CVP.png}
%\caption{Illustration of the closest vector problem (basis vectors in blue, external vector in green, closest vector in red). Source: Wikipedia (by Sebastian Schmittner)}
%\end{figure}

In \cite{micciancio}, it is proved that $\gamma$-SVP is NP-hard for $\gamma<\sqrt{2}$ and in \cite{boas}, it is proved that CVP is NP-complete, hence if $P\neq NP$, these two problems cannot be solved in polynomial time, even with the aid of a quantum computer.

\subsection{The Learning With Errors problem (LWE)}

Let $q$ be a rational prime for which a suitable choice will be made later. 

\begin{defn}The real torus of dimension $1$ is the quotient group $\mathbb{T}=\mathbb{R}/\mathbb{Z}$, its elements are equivalence classes of the form $x+\mathbb{Z}$ with $x\in [0,1)$.
\end{defn}

\begin{lem}The following map is a group monomorphism:
$$
\begin{array}{ccc}
\mathbb{F}_q & \hookrightarrow & \mathbb{T}\\
a+q\mathbb{Z} & \mapsto & \frac{a}{q}+\mathbb{Z}.
\end{array}
$$
%Here $a$ stands for a coset representative $0\leq a\leq q-1$.
\label{torus}
\end{lem}

A realization of lattice-based cryptography immune to all current quantum attacks and with a good chance of being NP-hard relies on the LWE problem, which we describe in this subsection.

\begin{defn}[LWE-oracles]

Let $\chi$ be a discrete random variable with values in $\mathbb{T}$. For $s\in \mathbb{F}_q^n$, chosen uniformly at random, a LWE-oracle with respect to $s$ and $\chi$ is a probabilistic algorithm $A_{s,\chi}$ which, at each execution, performs the following steps:
\begin{itemize}
\item[1. ] Samples a vector $a$ uniformly at random from $\mathbb{F}_q^n$.
\item[2. ] Computes the scalar product $\langle a,s \rangle$.
\item[3. ] Samples $e\in\mathbb{T}$ from $\chi$.
\item[4. ] Outputs the vector $\left[a, \frac{1}{q}\langle a,s\rangle+e\right]\in\mathbb{F}_q^n\times \mathbb{T}$.
\end{itemize}
\end{defn}

\begin{defn}[The LWE problem]

Let $\chi$ be a discrete random variable with values in $\mathbb{T}$ as before. The LWE problem for $\chi$ and $q$ is defined as follows:
\begin{itemize}
\item[a)] Search version: for an element $s\in \mathbb{F}_q^n$ chosen uniformly at random and a LWE-oracle $A_{s,\chi}$, if an adversary is given access to arbitrarily many samples of the LWE distribution, this adversary must recover $s$ with non-negligible advantage.
\item[a)] Decissional version: for an element $s\in \mathbb{F}_q^n$ chosen uniformly at random and a LWE oracle $A_{s,\chi}$, the adversary is asked to distinguish, with non-negligible advantage, between arbitrarily many samples from $A_{s,\chi}$ and the same number of samples $(a_i,b_i)\in\mathbb{F}_q^n\times\mathbb{T}$ where $a_i$ and $b_i$ are chosen independently and uniformly at random from $\mathbb{F}_q^n$ and $\mathbb{T}$.
\end{itemize}
\end{defn}

From now on, $\chi$ will be an $\mathbb{F}_q$-valued Gaussian variable, which can be thought of as having values on $\mathbb{T}$ via Lemma \ref{torus}. Such a variable is defined as follows:  For $\sigma, c\in\mathbb{R}$ we set $\rho_{\sigma,c}(x) = \exp{\frac{-(x-c)^2}{2 \sigma^2}}$. Write 
$$S_{\sigma,c} = \rho_{\sigma,c}(\mathbb{Z}) = \sum_{k=-\infty}^{\infty}\rho_{\sigma,c}(k),$$ 
and define $D_{\sigma,c}$ to be the distribution on $\mathbb{Z}$ such that the probability of $x\in\mathbb{Z}$ is $\rho_{\sigma,c}(x)/S_{\sigma,c}$. Finally, the discrete Gaussian distribution $\chi$ with values in $\mathbb{F}_q$, mean $0$, and parameter $\sigma$ is defined by the probability function
$$
Pr\left[\chi=k\right]=\sum_{n\equiv k\pmod{q}}pr\left[D_{\sigma,0}=n\right].
$$
Some words of caution: first, the variance of $\chi$ should be very close to $\sigma^2$, but not neccesarily must be equal: in lattice-based cryptography one speaks about discrete random variables of parameter (rather than variance) $\sigma^2$. Second, effective sampling from discrete Gaussian distributions is a difficult topic  and in practical cases it is approached only by numerical approximation (see \cite{statistics}). 

We conclude here with the following result due to Regev (\cite{regev}): a polynomial time quantum reduction from the SVP problem to the LWE problem, which backs the hardness of LWE and makes it a candidate to sustain a cryptosystem from it, as we will see in the next subsection.

\begin{thm}[Regev, \cite{regev}]Let $\chi_r$ be a discrete Gaussian of parameter $r^2$, $q$ a prime and $s\in\mathbb{F}_q^n$. Assume $r\geq 2\sqrt{n}$. Then, there is a quantum polynomial time reduction from $\gamma$-SVP, with $\gamma=\tilde{O}(nq/r)$to the search LWE problem attached to the LWE oracle $A_{s,\chi_r}$.
\label{regevth}
\end{thm}

\subsection{Attacks against LWE}
In the language of Machine Learning, due to Theorem \ref{regevth}, a training algorithm for the LWE problem can be turned, in polynomial time (on a quantum computer), into an algorithm (of the same complexity) which solves the SVP problem. If the $\gamma$-SVP problem were NP-hard for the value of $\gamma$ given in Theorem \ref{regevth}, it would follow the NP-hardness of the LWE problem. However, that value of $\gamma$ depends on $r^2$, the parameter of $\chi_r$, and the values of $r$ for which the LWE-problem for $A_{s,\chi_r}$ results in a correct cryptosystem is bigger than $\sqrt{2}$, the value for which SVP is NP-hard. Hence, Regev's reduction cannot be used to prove NP-hardness of SVP. Nevertheless, this kind of result can be seen as a clue towards its security. 

However, LWE has not been yet broken and there is a wide consense of the problem being \emph{intractable}. Nevertheless, some ad-hoc instantiations may be insecure against very simple attacks. Given $m$ LWE samples $\{(a_i,b_i=\frac{1}{q}\langle s,a_i\rangle+e_i)\}_{i=1}^m$, we can put them in columns to obtain a matrix $A=[a_1|...|a_m]\in\mathbb{F}_q^{n\times m}$ and set $\textbf{b}=\frac{1}{q}A^ts+\textbf{e}$, where $\textbf{e}$ is the column vector of errors. We analyze three vulnerable instantiations:
\begin{itemize}
\item[1. ] If $\chi$ is identically zero (errorless LWE), $s$ can be recovered via Gaussian elimination as long as the rows of A are linearly independent, which holds with high probability for $m>n$.
\item[2. ] If $\chi$ takes values in $z+[-1/2,1/2)$ with fixed $z\in\mathbb{R}$, we can round away each coordinate of $\textbf{b}$ and subtract $z$ to reduce to errorless LWE.
\item[3. ] If each group of $k$ samples has an error vector drawn from some distribution $\kappa$ in $\mathbb{R}^k$ and some discretized error coordinate is always $0$ under $\kappa$, we can ignore the samples corresponding to the other coordinates and since we have access to unlimited samples by hypothesis, we can equally reduce ourselves to errorless LWE. Analogously, we can reduce to errorless LWE if the sum (or a linear combination) of the $k$ error coordinates in each group is $0$.
\end{itemize}
\label{lwevulnerable}

\begin{rem}Generalizing Case 2 in the above analysis, the error distribution $\chi$ is said \emph{not to wrap around $\mathbb{Z}$} if $Pr_{e\leftarrow\chi}\left\{e\not\in z+[\frac{1}{2},\frac{1}{2}) \right\}$ is small enough for some known $z\in\mathbb{R}$. In this case, again by our unlimited access to the LWE oracle, the same attack as in Case 2 has good chance of success.
\end{rem}

Other instantiations of LWE can be attacked by more sophisticated means. For instance, as described in \cite{ag11}, if all the discretized errors in our samples (i.e. seen not in the torus but in $\mathbb{F}_q^n$, after rounding to the closest integer) lie in a known set of size $d$, then search LWE can be broken in approximately $n^d$ time and space, using $n^d$ samples. If $d=O(1)$, the attack is polynomial in the dimension, while if $d=n^{1-\varepsilon}$, the attack is sub-exponential. For details cf. \cite{peikert}, Section 2.

What these attacks should make us learn is that the distribution $\chi$ should be very carefully chosen, to avoid falling in a low dimensional subspace of $\mathbb{F}_q^n$, in which case, reduction to errorless LWE might have a good chance of success.

\subsection{The LWE cryptosystem}
Based on the hardness guarantee in Theorem \ref{regevth}, and avoiding the above problematic instantiations, the LWE problem can be used to build the following cryptosystem:

\begin{construction}[LWE cryptosystem, Regev (\cite{regev})]
\begin{itemize}
\item[ ]
\item[1. ] Parameters: $n,m\in\mathbb{N}$, $\alpha>0$.
\item[2. ] Private key: $s\in\mathbb{F}_q^n$ chosen uniformly at random.
\item[3. ] Public key: 

\begin{itemize}
	\item[3.1 ] Sample $a_1,...,a_m\in\mathbb{F}_q^n$, independently and uniformly at random.
	\item[3.2 ] Sample $e_1,...,e_m\in\mathbb{F}_q$, independently from $\chi$, which is assumed here to be a discrete Gaussian of zero mean and parameter $\frac{\alpha q}{2\pi}$.
	\item[3.3] Publish $\lbrace \left[ a_i,b_i=\langle a_i,s\rangle +e_i \right]\rbrace_{i=1}^m$.
\end{itemize}

\item[4. ] Encryption: for a bit $z\in\mathbb{F}_2$, consider it as an element of $\mathbb{F}_q$ by mapping the $0$ and $1$ of $\mathbb{F}_2$ to  the $0$ and $1$ of $\mathbb{F}_q$. Select a random subset $S\subseteq \{ 1,...,m\}$ and map
$$
z \mapsto \left[u,v\right] =\left[ \sum_{i\in S}a_i, z\lfloor \frac{q}{2} \rfloor + \sum_{i\in S}b_i\right].
$$
\item[5. ] Decryption: on receiving an encrypted message $\left[u,v\right]$, compute $d:=v-\langle u,s \rangle $. This equals $z\lfloor \frac{q}{2} \rfloor + \sum_{i\in S} e_i$. If $z=0$, then $d$ has absolute value below $\lfloor \frac{q}{4} \rfloor$ with probability as close to $1$ as desired, depending on how we choose the parameter $\alpha$. So, if this is the case,  decrypt to $0$,  otherwise, decrypt to $1$.
\end{itemize}
\end{construction}

The right choice of $q$, $m$ and $\alpha$ is given in the following result, whose proof is omited since it is very similar to the cryptographic scheme presented in the next subsection, whose proof we will discuss.

\begin{thm}If $q\in \{n^2,...,2n^2\}$, $\alpha=\frac{1}{\sqrt{n}\log^2(n)}$ and $m$ is of the order of $n \log(q)$, then the LWE cryptosystem is correct and pseudorandom\footnote{I.e. statistically indistinguishable from a uniform distribution.}.
\end{thm}

As we can see, a public key for LWE has $m$ vectors in $\mathbb{F}_q^n$, since $m$ is of the order of $n\log(n)$, it turns out that a public key has an $\mathbb{F}_q$-size of the order $n^2\log(n)$. This quadratic overhead is an unfeasible constrain from a practical point of view, in particular in settings such as hand-held digital broadcasting, mobile encryption and small devices in tentative applications of the IoT (Internet of Things), where the hardware has a relatively small memory. Moreover, in other recent scenarios where homomorphic encryption is desirable, LWE cannot fit well if the plaintext space is big enough. Such a scenario is that of electronic elections (e-voting and i-voting), which has to combine encryption with signature and authentication. For a large enough country, the size of the keys (which even if a pseudorandom generator is used, must grow with the size of the plaintext space) is certainly to be taken into account. 

A variation of the LWE problem, the ring learning with errors (RLWE) problem was introduced to tackle this quadratic overhead in the key sizes. The foundations of the problem require several notions from algebraic number theory, which we present next.

\section{Some basics of algebraic number theory}
Here we present the notions of algebraic number theory used to build the RLWE cryptosystem. Readers who are familiar with them can safely skip this section, since all our notations are standard. Readers who are not so familiar are referred to \cite{stewart}, Chapter 2 for more details.

\subsection{Algebraic number fields} An algebraic number field (number field, for short) is a field extension $K = \mathbb{Q}(\theta)/\mathbb{Q}$ of finite degree $n$, where $\theta$ satisfies a relation $f(\theta) = 0$ for some irreducible polynomial $f(x) \in\mathbb{Q}[x]$, which is monic without loss of generality. The polynomial $f$ is called the minimal polynomial of $\theta$, and $n$ is also the degree of $f$. Notice that $K$ is in particular an $n$-dimensional $\mathbb{Q}$-vector space and the set $\{1,\theta,...,\theta^{n-1}\}$ is a $\mathbb{Q}$-basis of $K$ called a power basis. Notice that associating $\theta$ with the unknown $x$ yields a natural isomorphism between $K$ and $\mathbb{Q}[x]/f(x)$.

Let $\overline{\mathbb{Q}}$ denote an algebraic closure of $\mathbb{Q}$ fixed from now on. A number field $K = \mathbb{Q}(\theta)$ of degree $n$ has exactly $n$ field embeddings (injective field homomorphisms) $\sigma_i: K \to \overline{\mathbb{Q}}$ fixing $\mathbb{Q}$. Each embedding $\sigma_i$ is determined by $\sigma_i(\theta)=\theta_i$, where $\{\theta_i\}_{i=1}^n$ are the different roots of $f$. The number field is said to be Galois if $K$ is the splitting field of $f$.

\begin{example}Denote by $\sqrt[3]{2}$ the unique real cubic root of $2$. The number field $K=\mathbb{Q}(\sqrt[3]{2})$ is not Galois: indeed, the other two roots of the minimal polynomial, $X^3-2$ do not belong to $K$. To make it Galois, we need to adjoin $\omega_3$, a non-real cubic root of $1$.
\end{example}

An embedding whose image lies in $\mathbb{R}$ (corresponding to a real root of $f$) is called a real embedding; otherwise it is called a complex embedding. Since complex roots of $f$ come in conjugate pairs, so do the complex embeddings. The number of real embeddings is denoted $s_1$ and the number of pairs of complex embeddings is denoted $s_2$, so we have $n=s_1+2s_2$. If $s_2=0$ (resp. $s_1=0$) $K$ is said to be totally real (resp. totally imaginary).

\begin{defn} The canonical embedding $\sigma: K\to \mathbb{R}^{s_1}\times\mathbb{C}^{2s_2}$ is then defined as 
$$
\sigma(x) = (\sigma_1(x),...,\sigma_n(x)).
$$ 
\end{defn}
%Note that $\sigma$ is a ring homomorphism from $K$ to $\mathbb{R}^{s_1}\times\mathbb{C}^{2s_2}$, where multiplication and addition in the latter are both component-wise.

\subsection{Algebraic integers} An algebraic integer is an element of $\overline{\mathbb{Q}}$ whose minimal polynomial over $\mathbb{Q}$ has integer coefficients. For a number field $K$ of degree $n$, let $\mathcal{O}_K\subset K$ denote the set of all algebraic integers in $K$. This set forms a ring under addition and multiplication in $K$ (\cite{stewart}, Theorem 2.9), called the ring of integers of $K$. It happens that $\mathcal{O}_K$ is a free $\mathbb{Z}$-module of rank $n$, i.e., it is the set of all $\mathbb{Z}$-linear combinations of some basis $\mathcal{B}=\{b_1,...,b_n\}\subset \mathcal{O}_K$ of $K$ (\cite{stewart}, Theorem 2.16). Such a set $\mathcal{B}$ is called an integral basis. 

%\begin{rem}If $K$ is totally real, then $(\mathcal{O}_K,\sigma|_{\mathcal{O}_K})$ is a lattice in $\mathbb{R}^n$. If not, since $\mathbb{R}^{s_1}\times\mathbb{C}^{2s_2}$ can be seen as $\mathbb{R}^{s_1+4s_2}$, the lattice is not full rank there any more. It is possible to adapt the definition of lattice but we won't do it here, and if we need the full rank condition we will assume that $K$ is totally real.
%\end{rem}

\begin{example}Let $n>1$ be an integer. The set of primitive $n$-th roots of unity (those of the form $\theta_k=exp(2\pi ik/n)$, with $1\leq k \leq n$ coprime to $n$) forms a multiplicative group of order $m=\phi(n)$. The $n$-th cyclotomic polynomial is
$$
\Phi_n(x)=\prod_{k\in\mathbb{Z}^*_n}(x-\theta_k).
$$
This is the minimal polynomial of $\theta_k$ for each $k$, so that $K=\mathbb{Q}(\theta_k)$ is a number field of degree $m$. It can be proved (\cite{stewart} Chap 3) that the ring of integers of $K$ is precisely $\mathbb{Z}[\theta]$ for each $\theta=\theta_k$, with $k\in\mathbb{Z}^*_n$.
\label{cycloexample}
\end{example}

\begin{defn}A number field $K$ such that $\mathcal{O}_K=\mathbb{Z}[\alpha]$ for some $\alpha\in \mathcal{O}_K$ is said to be monogenic.
\end{defn}

\begin{example}Let $d$ be a square-free integer. Consider the number field $\mathbb{Q}(\sqrt{d})$. It can be shown that the ring of integers of $K$ is $\mathbb{Z}[\sqrt{d}]$ if $d\not\equiv 1\pmod{4}$ and $\mathbb{Z}[\frac{1+\sqrt{d}}{2}]$ otherwise.
\label{quadexample}
\end{example}

\begin{defn}[Norm, trace and discriminant]For a number field $K$ of degree $n$, given an element $\alpha\in K$, its norm is defined as the product
\begin{equation}
N(\alpha)=\sigma_1(\alpha)\cdots\sigma_n(\alpha),
\end{equation}
and the trace is
\begin{equation}
Tr(\alpha)=\sigma_1(\alpha)+...+\sigma_n(\alpha).
\end{equation}

The discriminant of $K$, denoted $\Delta_K$ is the square of the determinant of the following matrix:
$$
\left(
\begin{array}{ccc}
\sigma_1(\theta_1) & ... & \sigma_n(\theta_1)\\
\vdots & \ddots & \vdots\\
\sigma_1(\theta_n) & ... & \sigma_n(\theta_n)\\
\end{array}
\right),
$$
where $\{\theta_1,...\theta_n\}$ is an integral basis of $\mathcal{O}_K$. Notice that since lattice base-change matrices are unimodular, the definition does not depend on the choice of the basis\footnote{In most algebraic number theory texts our $\Delta_k$ is called the minimal discriminant, since it is possible to define such a determinant for each $K$-basis (not necessarily integral). We will only consider integral bases and minimal discriminants.}.
\end{defn}

\begin{example}[\cite{was} Prop. 2.7]Let $K_n$ denote the $n$-th cyclotomic field. Then, the discriminant of $K$ equals
$$
\Delta_{K_n}=(-1)^{\phi(n)/2}\frac{n^{\phi(n)}}{\prod_{p\mid n}p^{\frac{\phi(n)}{p-1}}}.
$$
\label{cyclodisc}
\end{example}

Norm and trace and discriminant are rational numbers Moreover, they are integers when restricted to $\mathcal{O}_K$.

\subsection{Ideals and ideal lattices} Recall that an ideal of a ring $R$ is an additive subgroup $I\subseteq R$ such that for each $\alpha\in R$ and each $\beta\in I$, it is $\alpha\beta\in I$. For instance, for $d\equiv 1\pmod{4}$, the subring $\mathbb{Z}[\sqrt{d}]$ is not an ideal of the ring of integers, just a subring with finite index.

Unlike $\mathbb{Z}$, in the ring of integers $\mathcal{O}_K$ of a number field $K$, it is not true that every element $\alpha\in\mathcal{O}_K$ is a unique product, up to order and units, of different irreducible elements\footnote{An element $\alpha$ of a ring $R$ is irreducible if for any $\beta,\gamma\in R$ such that $\alpha=\beta\gamma$, either $\beta$ or $\gamma$ is a unit.}. For example, in $\mathbb{Z}[\sqrt{-6}]$, we have $6=2\cdot 3=\sqrt{-6}\cdot\sqrt{-6}$, where $2,3$ and $\sqrt{-6}$ are irreducible elements. However, this generalisation holds if we replace \emph{(irreducible) elements} by \emph{(prime) ideals}:

\begin{thm}[\cite{stewart}, Theorem 5.6]$\mathcal{O}_K$ is a Dedekind domain. In particular, for each ideal $I\subseteq\mathcal{O}_K$, there exist unique prime ideals $\frak{p}_1,...,\frak{p}_r$ and unique integers $e_1,...,e_r\in\mathbb{Z}_{\geq 0}$ such that
$$
I=\frak{p}_1^{e_1}...\frak{p}_r^{eg_r}.
$$
Moreover, denoting $f_i=|\mathcal{O}_K/\frak{p}_i|$, for $i=1,...,r$, it is
$$
n=e_1f_1+...+e_rf_r.
$$
\end{thm}
\begin{example}In $\mathbb{Z}[\sqrt{-17}]$, we can express the principal ideal $\langle 18 \rangle$ as the product $\frak{p}_1^2\frak{p}_2^2\frak{p}_3^2$, with $\frak{p}_1=\langle 2, 1+\sqrt{-17}\rangle$, $\frak{p}_2=\langle 3, 1+\sqrt{-17}\rangle$ and $\frak{p}_3=\langle 3, 1-\sqrt{-17}\rangle$.
\end{example}
\begin{defn}Let $p\in\mathbb{Z}$ be a rational prime decomposed as $(p)=\frak{p}_1^{e_1}...\frak{p}_r^{e_r}$ in $\mathcal{O}_K$ with $\frak{p}_i$ prime ideals. The number $e_i$ is called the ramification index of $p$ at $\frak{p}_i$ and if $e_i>1$, then $p$ is said to ramify at $\frak{p}_i$. The number $f_i=|\mathcal{O}_K/\frak{p}_i|$ is called inertia degree of $p$ at $\frak{p}_i$. If $r=n$, then all the $e_i$ and $f_i$ equal $1$ and $p$ is said to be totally split.
\end{defn}

A theorem by Minkowski states that every number field has only finitely many ramifying primes, which are precisely the rational primes dividing the discriminant. Hence, going back to Example \ref{cyclodisc}, we see that for the $n-th$ cyclotomic field the ramifying primes are those which divide $n$.

\begin{defn}Let $R$ be a discrete ring (free and finitely generated as abelian group) and $\sigma:R\to\mathbb{R}^n$ an additive monomorphism. Nottice that $\sigma(R)$ is a lattice. The family of ideal lattices (for the ring $R$ and embedding $\sigma$) is the set of all lattices $\sigma(I)$ for ideals $I$ in $R$. 
\end{defn}
For instance, for $R=\mathbb{Z}[x]/f(x)$, the coefficient embedding maps any element of R to the integer vector in $\mathbb{Z}^n$ whose coordinates are exactly the coefficients of that element when viewed as a polynomial residue. When $R=\mathcal{O}_K$, the canonical embedding $\sigma$ provides in a natural way an ideal lattice for each ideal $I$ of $R$. 

Notice that for the canonical embedding, multiplication and addition are preserved componentwise. On the contrary, for instance, for the ring $R=\mathbb{Z}[x]/(x^n+1)$, the componentwise multiplication in $\mathbb{Z}_q^n$ doesn't correspond with multiplication in $R$: multiplying by $x$, is equivalent to shifting the coordinates and negate the independent term. This is one of the advantages of using the canonical embedding.

Moreover, one has the following connection between the fundamental parallelotope of $\sigma(\mathcal{O}_K)$ and the discriminant $\Delta_K$:

\begin{thm}[\cite{stewart}, cf. Theorem 8.1] Assume that the number field $K$ has $s$ pairs of complex embeddings. Then, the Euclidean measure of the fundamental parallelotope of $\sigma(\mathcal{O}_K)$ equals $2^s\sqrt{\Delta_K}$.
\end{thm}

\section{Ring learning with errors: problems, cryptosystem and key exchange}
To define the ring learning with errors problem (RLWE), let $K$ be a number field of degree $n$ and ring of integers $\mathcal{O}_K$, regarded as a lattice in $\mathbb{R}^n$, by means of the canonical embedding. Closely connected with RLWE is the polynomial learning with errors problem (PLWE). Next we formally introduce both problems and explore their relation. 
\subsection{Statement of the problems}
In the rest of this subsection $f(x)\in\mathbb{Z}[x]$ is supposed to be a monic irreducible polynomial of degree $n$ and  $q$ is a rational prime which we will choose later. Define, further, $\mathcal{O}:=\mathbb{Z}[x]/(f(x))$, which can also be regarded as a lattice in $\mathbb{R}^n$ by means of the coordinate embedding

\begin{equation}
\begin{array}{ccc}
\sigma: \mathcal{O} & \to & \mathbb{R}^n\\
\displaystyle\sum_{i=0}^{n-1}a_i\overline{x}^i & \mapsto & (a_0,...,a_{n-1}).
\end{array}
\end{equation}
Each root $\alpha$ of $f$ defines a number field $K_{\alpha}=\mathbb{Q}(\alpha)$. Moreover, the ring $\mathbb{Z}[\alpha]$ is a finite index suborder of the ring of integers $\mathcal{O}_{K_{\alpha}}$. The restriction of the canonical embedding to $\mathbb{Z}[\alpha]$ also provides a lattice in $\mathbb{R}^n$. A very common choice is $f(x)=x^{2^k}+1$, the $2^{k+1}$-th cyclotomic polynomial (cf. \cite{stehle2}).

The $n$-dimensional torus attached to $\mathcal{O}_K$ is $\mathbb{T}:=(K\otimes_{\mathbb{Q}}\mathbb{R})/\mathcal{O}_K$, and the $f$-torus is defined to be $\mathbb{T}_f:=\mathbb{R}_q[X]/(f(X))$, with $\mathbb{R}_q:=\mathbb{R}/\mathbb{Z}$. As in Lemma \ref{torus}, there are embeddings $\mathcal{O}_K/q\mathcal{O}_K\hookrightarrow\mathbb{T}$ and $\mathcal{O}/q\mathcal{O}\hookrightarrow\mathbb{T}_f$.

\begin{defn}[RLWE and PLWE-oracles]

\begin{itemize}
\item[]
\item[1.] Let $\chi$ be a discrete random variable with values in $\mathcal{O}_K/q\mathcal{O}_K$ (which we regard as taking values in $\mathbb{T}$). For $s\in \mathcal{O}_K/q\mathcal{O}_K$ chosen uniformly at random, a RLWE-oracle with respect to $s$ and $\chi$ is a probabilistic algorithm $A_{s,\chi}$ which, at each execution performs the following steps:
\begin{itemize}
\item[1. ] Samples an element $a\in \mathcal{O}_K/q\mathcal{O}_K$ uniformly at random,
\item[2. ] Samples an element $e$ from $\chi$,
\item[3. ] Outputs the pair $(a,as+e)\in \mathcal{O}_K/q\mathcal{O}_K\times\mathbb{T}$.
\end{itemize}
 
\item[2.] Let $f(x)\in\mathbb{Z}[x]$  be monic irreducible as above and $\chi$ a discrete random variable  with values in $\mathcal{O}/q\mathcal{O}$ (which we regard as taking values in $\mathbb{T}_f$) . For $s\in \mathcal{O}/q\mathcal{O}$ chosen uniformly at random, a PLWE-oracle with respect to $s$ and $\chi$ is a probabilistic algorithm $A_{s,\chi}$ which, at each execution performs:
\begin{itemize}
\item[1. ] Samples an element $a\in \mathcal{O}/q\mathcal{O}$ uniformly at random,
\item[2. ] Samples an element $e$ from $\chi$,
\item[3. ] Outputs the pair $(a,as+e)\in \mathcal{O}/q\mathcal{O}\times\mathbb{T}_f$.
\end{itemize}

\end{itemize}
\end{defn}

\begin{defn}[The RLWE/PLWE problem]

Let $\chi$ be a discrete random variable with values in $\mathcal{O}_K/q\mathcal{O}_K$ (in $\mathcal{O}/q\mathcal{O}$). The RLWE (PLWE) problem for $\chi$ is defined as follows:
\begin{itemize}
\item[a)] Search version: for an element $s\in \mathcal{O}_K/q\mathcal{O}_K$ ($\mathcal{O}/q\mathcal{O}$) chosen uniformly at random and a RLWE (PLWE)-oracle $A_{s,\chi}$, if an adversary is given access to arbitrarily many samples $(a_i,a_is+e_i)$ of the RLWE (PLWE) distribution, this adversary must recover $s$ with non-negligible advantage.
\item[a)] Decisional version: for an element $s\in \mathcal{O}_K/q\mathcal{O}_K$ ($\mathcal{O}/q\mathcal{O}$) chosen uniformly at random and a RLWE (PLWE)-oracle $A_{s,\chi}$, the adversary is asked to distinguish, with non-negligible advantage, between arbitrarily many samples from $A_{s,\chi}$ and the same number of samples $(a_i,b_i)$, taken uniformly at randon from $\mathcal{O}_K/q\mathcal{O}_K\times\mathbb{T}$ ($\mathcal{O}/q\mathcal{O}\times\mathbb{T}_f$).
\end{itemize}
\end{defn}

Some words on the class of distributions we will use from now: first, notice that if $q$ is totally split, what we will frequently assume, a RLWE-sample can be seen as an $n$-tuple of coordinates with values in $\mathbb{F}_q$. However, such a RLWE-sample is indeed \emph{much more} than $n$ LWE-samples: $\mathcal{O}_K/q\mathcal{O}_K$ is not only an $\mathbb{F}_q$- vector space; it also has a ring structure. The flexibility and power of RLWE comes from exploiting the ring structure instead of the sheer lattice structure. This is the reason why instead of taking $n$-independent discrete one-dimensional Gaussians, we rather use an $n$-dimensional one. 

As in the $1$-dimensional case, the mean will also be supposed $0$, but in the RLWE scenario, the variance-covariance matrix (or rather, the multidmensional parameter) is normallly chosen, depending on the application, a) either to be diagonal, which is referred to as saying that the distribution is elliptic\footnote{This is useful when carrying out security-reduction proofs.}, or b) to have  the diagonal elements bounded in absolute value by $\alpha n^{1/4}$, for $\alpha$ a parameter which will be made explicit in the next theorem, which backs the security of the decisional RLWE-problem (hence of the search RLWE-problem) in the security of the SVP over ideal lattices.

%\begin{figure}
%  \centering
%   \includegraphics[width=0.9\textwidth]{ndb.png}
%\caption{Discrete bivariate normal distribution supported on a lattice. Source: \cite{regev}, with permission given by the author.}
%\end{figure}

Hence, from now on, we assume that $\chi_{\alpha}$ is an elliptic $n$-dimensional discrete $\mathbb{T}$-valued Gaussian of $0$-mean and the elements of the diagonal are bounded as explained. The details are delicate and can be omited in a first study, since the aforementioned bound is what really matters for most proofs, but the reader is referred to \cite{LPR}, p. 19 for more information.

\begin{thm}[\cite{LPR}, page 19]Let $K$ be the $m$-th cyclotomic number field of degree $n=\phi(m)$ and $R=\mathcal{O}_K$ its ring of integers. Let $\alpha<\sqrt{\log n/n}$ and let $q=q(n)\geq 2$, $q\equiv 1\pmod{m}$ be a prime bounded by a polynomial in $n$ such that $\alpha q\geq \omega(\sqrt{\log n})$\footnote{A function $f: \mathbb{N}\to\mathbb{R}$ is $\omega(g)$, for $g:\mathbb{ N}\to\mathbb{R}$ (denoted as $f=\omega(g)$) if for each integer $k>0$ there exists an integer $n_0>0$ such that  for each $n\geq n_0$, it is $|f(n)|\geq k |g(n)|$. The notation $f(n)\geq\omega(g(n))$ means that the asymptotic behaviour of $f$  is at least as fast as $\omega(g(n))$ .}. There is a polynomial time quantum reduction from $\tilde{O}(\sqrt{n}/\alpha)$-SVP on ideal lattices of $K$ to the decisional RLWE problem for $K$ and $\chi_{\alpha}$.
\label{thm1}
\end{thm}

The proof consists of two parts: the first is a quantum reduction from worst case approximate SVP on ideal lattices to the search version of RLWE. The reduction works in general, not for just cyclotomic number fields. It uses the iterative quantum reduction for general lattices in \cite{regev} as a black box, the main effort being the classical (non-quantum) part, which requires a careful handling of the canonical embedding and a smart use of the Chinese Remainder Theorem.

The second part shows that the RLWE distribution is pseudorandom via a classical reduction from the search version, which has been shown at least as hard as SVP for ideal lattices in the first part. It uses the fact that the cyclotomic field is Galois and the fact that $q\equiv 1\pmod {n}$, namely, that the ideal $ qR$ splits totally into $n$ different prime ideals in $R$.

In \cite{PRS} Theorem 6.2, the authors build on the same number-theoretical kind of arguments as in \cite{LPR} to prove an analogue of Theorem \ref{thm1} for non-cyclotomic Galois number fields.

\subsection{Equivalence between formulations}

In \cite{LPR}, the RLWE problem is introduced via $\mathcal{O}^{\vee}_K/q\mathcal{O}^{\vee}_K$  as sample space, instead of $\mathcal{O}_K/q\mathcal{O}_K$, where $\mathcal{O}^{\vee}_K$ means the dual of $\mathcal{O}_K$ with respect to the trace map, namely:
$$
\mathcal{O}^{\vee}_K=\{\alpha\in K: Tr(\alpha)\in\mathbb{Z}\}.
$$
We have avoided this formulation to spare the definition of the different ideal and, no less important, for the sake of the extension of our presentation. In any case, both formulations are equivalent (\cite{RSW} Theorem 2.13). By equivalence we mean that every solution for primal-RLWE can be turned in polynomial time into a solution for dual-RLWE (and viceversa, but this is immediate, since $\mathcal{O}_K\subseteq\mathcal{O}^{\vee}_K$), incurring in a noise increase which is polynomial in the number field degree.

Before speaking about the RLWE/PLWE equivalence we need to introduce a key concept: the condition number, which measures the distortion between the lattices given by the canonical embedding and the coordinate embeding. Let's do that.

For a monic irreducible polynomial of degree $n$, $f(x)\in \mathbb{Z}[x]$ and $\theta$ a root of $f(x)$, consider again the subring $\mathbb{Z}[x]/(f(x))\cong\mathbb{Z}[\theta]\subseteq\mathcal{O}_K$. As lattices, $\mathbb{Z}[x]/(f(x))$ is endowed with the coordinate embedding while $\mathbb{Z}[\theta]$ is endowed with the canonical embedding inherited from $\mathcal{O}_K$, and the evaluation-at-$\theta$ morphism causes a distortion between both. Explicitly, the transformation between the embeddings is given by
\begin{equation}
\begin{array}{ccc}
V_f: \mathbb{Z}[x]/(f(x)) & \to & \sigma_1(\mathcal{O}_{K})\times\cdots\times\sigma_n(\mathcal{O}_{K})\\
\displaystyle\sum_{i=0}^{n-1}a_i\overline{x}^i & \mapsto & \left(\begin{array}{cccc}1 & \theta_1 & \cdots & \theta_1^{n-1}\\ 
1 & \theta_2 & \cdots & \theta_2^{n-1}\\  
\vdots & \vdots & \ddots \vdots\\ 
1 & \theta_n & \cdots & \theta_n^{n-1}\end{array}\right)\left(\begin{array}{c}a_0 \\ a_1 \\  \vdots \\a_{n-1}\end{array}\right),
\end{array}
\label{latticebij}
\end{equation}
where $\theta=\theta_1,\theta_2,...,\theta_n$ are the Galois conjugates of $\theta$. As we see, the transformation $V_f$ is given by a Vandermonde matrix.

For any matrix $A=(a_{ij})\in M_{n\times n}(\mathbb{C})$, denote its transposed conjugate by $A^*$. The Frobenius norm is defined as 
\begin{equation}
||A||:=\sqrt{Tr(AA^*)}=\sqrt{\sum_{i,j=1}^n|a_{ij}|^2}.
\label{easy}
\end{equation}
The noise provoked by $V_f$ will remain \emph{controlled} whenever $||V_f||$ and $||V_f^{-1}||$ remain so, and the product $||V_f||||V_f^{-1}||$ serves as a reasonable measure of this control (cf. \cite{RSW} Ch. 4).
\begin{defn}The condition number of an invertible matrix $A\in\mathrm{M}_n(\mathbb{C})$ is defined as Cond$(A):=||A|||A^{-1}||$.
\end{defn}
Thus, in the monogenic case, the problem of the equivalence is the problem of showing that $Cond(V_f)=O(n^r)$ for some $r$ independent of $n$. The non-monogenic case needs an intermediate reduction that we will not address here.

In the above mentioned paper \cite{RSW}, the authors introduce the framework to study the RLWE/PLWE-equivalence in general and prove it for the following family of polynomials:

\begin{thm}[\cite{RSW}, pag. 4 and Theorem 4.7] There is a polynomial time reduction algorithm from RLWE over $K_{f_{n,p}}$ to PLWE for $f_{n,p}(x)$ where $K_{f_{n,p}}$ is the splitting field of $f_{n,p}(x)=x^n+xp(x)-r$ where $n\geq 1$, $p(x)$ runs over polynomials with $deg(p(x))<n/2$ and $r$ runs over primes such that $25||p||_1^2\leq r\leq s(n)$, with $s(x)$ a polynomial. Notice that there is a trivial reduction from PLWE to RLWE.\footnote{For $p(x)=\displaystyle\sum_{i=0}^np_ix^i\in\mathbb{R}[x]$, the $1$-norm is defined as $||p||_1=\displaystyle\sum_{i=0}^n|p_i|$}
\label{rswfam}
\end{thm}
The argument to prove this theorem is, first, to consider the family of polynomials $\phi_{n,a}(x):=x^n-a$, with $a\in\mathbb{Z}\setminus\{0\}$ square-free. Denoting by $K_{\phi_{n,a}}$ the splitting field of $\phi_{n,a}(x):=x^n-a$, the authors check in first place the equivalence for $K_{\phi_{n,a}}$ and they show, via a careful use of Rouch\'e theorem, that when $\phi_{n,a}(x)$ is perturbed by adding another polynomial with degree smaller than $n/2$ the roots of both polynomials are close enough.  

A reason to be interested in such an equivalence is that working with polynomial rings instead of rings of integers of number fields is more amenable for computer implementations. In \cite{bernstein}, it is shown how the arithmetic of several polynomial rings leads to very efficient cryptographic designs.

\subsection{The cyclotomic case}

In practice, the number fields we are the most interested in  cryptography are the cyclotomic number fields: they are very well understood and enjoy very nice arithmetic guatantees, like monogeneicity, which allows an amenable and efficient use for implementations. However, until recently, very little was known regarding the equivalence, apart from the power-of-two case:  the ideas in \cite{DD} can be applied to show the equivalence for cyclotomic number fields of degree $2^kp$ or $2^kpq$ with $p,q$ primes and $q<p$.  Besides that, the family in Theorem \ref{rswfam} is somehow artificially constructed, but, some of the ideas have been used by this author to give a partial proof of the equivalence in the cyclotomic case (\cite{blanco}). This proof is, to our knowledge, the first given for general cyclotomic degree (but with the caveat of imposing a condition which we comment next).

Before that, let us examine first the power-of-two degree. 

\begin{thm}Let $n=2^k$ and $m=\phi(n)=2^{k-1}$. Then, the map $V_{\Phi_n}$ is a scaled isommetry. In addition, $Cond(V_{\Phi_n})=m$.
\end{thm}
\begin{proof}To see that $V_{\Phi_n}$ is a scaled isometry, observe that when we multiply $V_{\Phi_n}$ by its conjugate transposed, the elements over the diagonal in the product matrix are identically $m$, and outside the diagonal, the element in position $(i,j)$ in the product matrix equals
$$
\sum_{k=0}^{m-1}\zeta_i^k\overline{\zeta_j}^k=\frac{1-\zeta_i^m\overline{\zeta_j}^m}{1-\zeta_i\overline{\zeta_j}}.
$$
But since $\zeta_i$ are $n$-primitive roots (and so are $\overline{\zeta_i}$), then $\zeta_i^m=-1$ and the sum vanishes. Hence, we have that
$$
V_{\Phi_n}V_{\Phi_n}^*=mId,
$$
and $m^{-1/2}V_{\Phi_n}$ is an isometry. For the condition number, we write $V_{\Phi_n}^{-1}=m^{-1}V_{\Phi_n}^*$, hence $||V_{\Phi_n}^{-1}||=1$. By Lemma \ref{easy}, the result follows.
\end{proof}

The main result in \cite{blanco} is a polynomial bound on the condition number for cyclotomic number fields which only depends on a) the number of different primes dividing the conductor and b) the degree of the number field, and what is more important, the dependence on the degree is polynomial once the number of different prime divisors has been fixed. Let us see how.

For $n\geq 1$, denote by $rad(n)$ the product of all the different primes dividing $n$ (without exponents). For the $n$-th cyclotomic polynomial $\Phi_n(x)$, denote by $A(n)$ the maximum of all the coefficients in absolute value. For instance, for $n=p^r$, prime, $A(n)=1$, and for $n=pq$, with $p,q$ prime, all the coefficients are $0,\pm 1$, due to a classical result by Migotti, hence $A(n)=1$. Our result is as follows:

\begin{thm}[\cite{blanco} Thm. 3.10] Let $n\geq 1$ and $m=\phi(n)$. If $rad(n)=p_1...p_k$, then:
$$
Cond(V_{\Phi_n})\leq 2rad(n)n^{2^k+k+2}A(n).
$$
\label{main2}
\end{thm}
\begin{proof}First, from the very definition, one has $||V_{\Phi_n}||=m$. Second, we use the following identity, a proof of which can be found, for instance, in \cite{was} Ch. 1:
\begin{equation}
\Phi_n(x)=\Phi_{rad(n)}(x^{\frac{n}{rad(n)}}),
\end{equation}
which yields $A(n)=A(rad(n))$. The technical core of the result is a series of upper bounds for the entries $w_{ij}$ of the inverse matrix $V_{\Phi_n}^{-1}$, of which the most important is:
$$
|w_{ij}|\leq 2rad(n)n^{2^k+k}A(n).
$$
\end{proof}
Now, to obtain the polynomial bound, we need to bound $A(rad(n))$, which we do with the aid of a classical result due to Bateman:

\begin{thm}[Bateman, \cite{bateman2}] Let $n=p_1...p_k$ with $p_1<...<p_k$. Then
$$
A(n)\leq  n^{2^{k-1}}. 
$$
\label{thacta}
\end{thm}
We can now derive the polynomial bound:
\begin{cor}[\cite{blanco} Cor. 3.11] Let $k\geq 1$ be fixed. If $n$ is the product of at most $k$ different primes, then $Cond(V_{\Phi_n})$ is polynomial in $n$. More in general, let $\mathcal{F}_k$ be a family of cyclotomic polynomials whose degree is divisible by at most $k$ different primes.  Assume that $A(n)=\mathcal{O}(n^r)$ for polynomials in $\mathcal{F}_k$. Then,
$$
Cond(V_{\Phi_n})=\mathcal{O}(n^{2^k+k+3+r}).
$$
\label{general}
\end{cor}
In \cite{blanco}, we also give a subexponential upperbound for the condition number if we do not fix the number of primes as well as more precise upper bounds for conductor divisible up to three primes. Namely:
\begin{thm}For $n\geq 1$ and $m=\phi(n)$, the following bounds hold for the condition number of cyclotomic polynomial $\Phi_n(x)$:
\begin{itemize}
\item[a)] (\cite{blanco} Thm. 4.1) If $n=p^k$ then $$
Cond(V_{\Phi_n})\leq 4(p-1)m.
$$

\item[b)] (\cite{blanco} Thm. 4.3) If $n=n=p^rq^l$ then
$$
Cond(V_{\Phi_n})\leq 2\phi(rad(n))m^2.
$$
\item[c)] (\cite{blanco} Thm. 4.6) If $n=p^lq^sr^t$ then 
$$
Cond(V_{\Phi_n})\leq 2\phi(rad(n))^2m^2.
$$
\end{itemize}
\end{thm}
In our proofs, apart from some of the ideas from \cite{RSW}, and some properties from cyclotomic polynomials from \cite{was}, we have used results from analytic number theory like the aforementioned Theorem \ref{thacta} due to Bateman and for the case of two and three primes, results by Migotti and Bang (\cite{bang}). This should highlight the strong link between ring lattice-based cryptography and number theory.

\subsection{The LPR (Lyubashevsky, Peikert and Regev) RLWE-cryptosystem}

Both RLWE and PLWE problems can be turned into public key cryptosystems, as we show next. We will focus in the PLWE version here. So, let $f(x)\in\mathbb{Z}[x]$ be a monic irreducible polynomial, $q$ a prime and set $\mathcal{O}=\mathbb{Z}[x]/f(x)$. Let $\chi_{\alpha}$ be an $\mathcal{O}/q\mathcal{O}$-valued discrete Gaussian (seen as taking values on $\mathbb{T}_f$) and as explained in the former subsection, we assume the parameter of $\chi_{\alpha}$ upper bounded entry-wise by $\alpha n^{1/4}$ with $\alpha\leq \sqrt{log(n)/n}$ and $q=q(n)$ as in \ref{thm1}. Take $n$ big enough so that $6\sqrt{log(n)/\sqrt{n}}<\frac{q}{4}$ (this will be used to grant the correctness of the cryptosystem).

\begin{construction}[The PLWE cryptosystem]

\begin{itemize}
\item[ ]
\item[1.] Key generation: choose $a\in\mathcal{O}/q\mathcal{O}$ uniformly at random and choose $s,e$ sampled from $\chi_{\alpha}$. The secret key will be $s$ and the public key will be the pair $(a,b=as+e)$.
\item[2.] Encryption: take a plaintext $z$ consisting of a stream of bits and regard it as a polynomial in $\mathcal{O}/q\mathcal{O}$, mapping each bit to a coefficient. Choose $r,e_1,e_2$ sampled from $\chi_{\alpha}$. Set $u=ar+e_1$ and $v=br+e_2+ \lfloor \frac{q}{2} \rfloor z$. The cyphertext is $(u,v)$.
\item[3.] Decryption: On cyphertext $(u,v)$, perform $v-us=er+e_2-e_1s+\lfloor \frac{q}{2} \rfloor z$ and round each coefficient either to zero or to $\lfloor \frac{q}{2} \rfloor$, whichever is closest mod $q$.
\end{itemize}
\end{construction}

\begin{prop}The PLWE cryptosystem is correct (i.e. decryption undoes encryption) and pseudorandom.
\end{prop}
\begin{proof}
For correctness, notice that for the chosen values of $\alpha$, $q$ and $n$, with arbitrarily large probability, the absolute values of the coefficients of $er+e_2-e_1s$ will be below $6\sqrt{log(n)/\sqrt{n}}<q/4$, so each bit of $z$ can be recovered by checking if its position in $v-us$ is less than $\lfloor q/4 \rfloor$, in which case, we decrypt it as $0$, and otherwise as $1$, as in the LWE scheme.

For pseudorandomness, first note that RLWE samples are pseudorandom even when $s$ is sampled from $\chi_{\alpha}$, by a transformation to the Hermite normal form. Therefore, public keys $(a,b)$ are pseudorandom and we can replace them by a uniform pair in $\mathcal{O}\times\mathbb{T}_f$. The observations of a passive adversary are $(a,u)$ and $(b,v)$ which are also pseudorandom, since $r$ is also sampled from $\chi_{\alpha}$.
\end{proof}

\begin{example}For around 100 bits security, current implementations use a parameter set with number field degree $n=256$, a 13-bit prime modulus $q$ and a narrow discrete Gaussian distribution with diagonal entries upper-bounded by 4.5.
\end{example}
%Among all arithmetic operations required in RLWE, polynomial multiplication is the most expensive and to perform it in an effective manner, it is used the NTT (Number Theoretic Transform), a variant of the Fourier Fast Transform over rings of integers (see \cite{speed} for a discussion of the topic and some improvements for special moduli).

\subsection{A RLWE-based key exchange protocol} Next we present a key exchange protocol based on RLWE and due to Ding (\cite{dingkex}). Earlier protocols for key transport were proposed by Peikert (\cite{peikertkex}) in 2012 and by Zhang in 2014. This protocol takes place between two devices typically called \emph{initiator} and \emph{respondent}, which we will call Alice and Bob respectively, for the sake of tradition, both of which have access to a discrete Gaussian of parameter $\alpha$ and both of which know $m=\phi(n)$, a prime $q$, the $n$-th cyclotomic polynomial $\Phi_n(x)$, hence the rings $\mathcal{O}=\mathbb{Z}[x]/(\Phi_x(x))$ and $\mathcal{O}/q\mathcal{O}$, and another polynomial $a(x)\in\mathcal{O}/q\mathcal{O}$. These data can and must be assumed to be publically known. The algorithm uses the following two functions:

\begin{defn}[Signalling and binary deletion functions] 
\hfill \break
Let $E:=\{-\lfloor\frac{q}{4}\rfloor,...,0,...,\lfloor\frac{q}{4}\rfloor\}$. The signalling function, denoted $Sig$, is the characteristic function of $\mathbb{F}_q\setminus E$, namely $Sig(v)=0$ if and only if $v\in E$, otherwise $Sig(v)=1$. The  binary deletion function is defined as 
$$
\begin{array}{ccc}
Mod_2: \mathbb{F}_q^2 & \longrightarrow & \mathbb{F}_2\\
(v,w) & \mapsto & v-\frac{w}{2}\pmod{2}.
\end{array}
$$
\end{defn}
The signalling function \emph{signals} the elements of $E$ as \emph{small}, returning $0$, while the binary deletion function returns $0$ on pairs $(v,w=Sig(v))$ corresponding to error bits, which belong to $E$ (i.e. $w=0$ and $v=2k$). The steps of the protocol are as follows:

\begin{itemize}
\item[1. ] Alice initiates:
\begin{itemize}
\item[1.1 ] Generates two polynomials $s_A$ and $e_A$ from the discrete Gaussian distribution $\chi_{\alpha}$.
\item[1.2. ] Computes $p_A=as_A+2e_A$.
\item[1.3. ] Sends Bob the polynomial $p_A$.
\end{itemize}
\item[2. ] Bob responds:
\begin{itemize}
\item[2.1 ] Generates two polynomials $s_B$ and $e_B$ from the discrete Gaussian distribution $\chi_{\alpha}$.
\item[2.2. ] Computes $p_B=as_B+2e_B$.
\item[2.3. ] Generates $e'_B$ from $\chi_{\alpha}$ and computes 
$$
k_B=p_As_B+2e'_B=as_As_B+2e_As_B+2e'_B.
$$
\item[2.4. ] Uses the signalling function to find $w=Sig(k_B)$ (applying $Sig$ coefficientwise to $k_B$)
\item[2.5 ] Performs $sk_B=Mod_2(k_B,w)$
\item[2. 6] Sends Alice $(p_B,w)$.
\end{itemize}
\item[3. ] Alice finishes:
\begin{itemize}
\item[3.1 ] Generates $e'_A$ from $\chi_{\alpha}$.
\item[3.2. ] Computes 
$$
k_A=p_Bs_A+2e_A'=as_As_B+2e_Bs_A+2e_A'
$$
\item[3.3. ] Alice performs $sk_A=Mod_2(k_A,w)$.
\end{itemize}
\end{itemize}
Notice that the elements $k_A$ and $k_B$ are only approximately equal, up to even errors, which allows the fuction $Mod_2$ to detect them. The $Sig$ function indicates the region in which each coefficient of a polynomial lies and helps to make sure that the error terms in $k_A$ and $k_B$  do not result in different mod $q$ operations. 

With a carefull choice of the parameter $\alpha$, it will be $sk_A=sk_B$ with overwhelming probability. The difficulty of breaking this scheme is that from $p_A$ and/or $p_B$, which is the only thing which a passive adversary is supposed to see, to recover $s_A$ and $s_B$, the adversary must break PLWE.
\begin{rem}In November 2015, Alkim, Ducas, P\"opplemann, and Schwabe recommended the parameters $n = 1024$ and $q =12289$ (see \cite{ADPS}). This represents a significant reduction in public key size over previous schemes, and was submitted to NIST with the name of NewHope. At the time of writing, NewHope has passed unbroken to the second round (see Section 7.3).
\end{rem}

\section{Attacks on RLWE}Detailed reports on the state of the art of attacks on the RLWE cryptosystem can be found in \cite{lauter} and \cite{peikert}. In \cite{lauter} the authors discuss a list of open questions in algebraic number theory motivated by several attacks on RLWE. This interplay between cryptography and number theory constitutes a fruitful link which is expected to motivate a flow of results from each direction to the other. %In particular, it is a good example about how a problem in number theory as the Artin conjecture on primitive roots and various generalisations can become relevant in the applied mathematics realm.

On the other hand, in \cite{peikert}, a comprehensive review of the known attacks and vulnerable instantiations is carried out from a geometric viewpoint. We present, at our introductory level, only a few of these attacks and questions, working out some details. Within this subsection, we assume as usual that $K$ is a number field of degree $n$, and in the PLWE setting, that the defining polynomial $f(x)$ splits totally over $\mathbb{F}_q[x]$. This is unnecessary but it will simplify the exposition, while keeping the essential facts.
\subsection{Reduction to LWE}
Let $\mathcal{B}$ be a $\mathbb{Z}$-basis of $\mathcal{O}$ such that its reduction modulo $q$, $\overline{\mathcal{B}}$, is an $\mathbb{F}_q$-basis of $\mathcal{O}/q\mathcal{O}$. Given $a\in\mathcal{O}/q\mathcal{O}$, multiplication by $a$ is an $\mathbb{F}_q$-linear map described by a matrix $A_a\in\mathbb{F}_q^{n\times n}$ with respect to $\overline{\mathcal{B}}$. Hence, a public key $(a,b=as+e)$ has attached the pair $(A_a,\textbf{b}=A_a\textbf{s}+\textbf{e})$, where $\textbf{s}$ and $\textbf{e}$ are, respectively, the coordinates of $s$ and $e$ with respect to $\overline{\mathcal{B}}$, which implies that one RLWE sample carries $n$ LWE samples.
A first attack is based on Case 2 in Section 3.2: if the $j$-th error coordinate with respect to $\overline{\mathcal{B}}$ does not wrap around $\mathbb{Z}$, namely, if $Pr_{e_j\leftarrow\chi}\left\{e_j\not\in [\frac{1}{2},\frac{1}{2}) \right\}$ is small enough, we have errorless LWE in the $j$-th row of $A_a$, and with enough samples we can recover $s$ with high probability.

Let now $\frak{q}\subseteq\mathcal{O}_K$ a prime ideal above $q$ of norm $N(\frak{q})=|\mathcal{O}_K/\frak{q}|$ and let $\chi$ be a Gaussian distribution over $K_{\mathbb{R}}=K\otimes\mathbb{R}$. Given RLWE samples $(a,b=as+e)$ where $a\in \mathcal{O}_K/q\mathcal{O}_K$ and $e$ taken from $\chi$, we can reduce them modulo $\frak{q}$ to obtain samples $(a'=a\pmod{\frak{q}}, b'=b\pmod{\frak{q}})$, with $b'=s'a'+e\pmod{\frak{q}}$ with $s'=s\pmod{\frak{q}}$, hence the secret now lies in a set of size $N(\frak{q})$. The following analysis is due to Peikert (cf. \cite{peikert} Section 3.2) and yields a potentially successful attack when $N(\frak{q})$ is not too large:

\begin{itemize}
\item[1. ]Since reduction modulo $\frak{q}$ takes uniform samples onto uniform samples, if $\chi$ modulo $\frak{q}$ is detectably non-uniform, we have an attack against decission RLWE.
\item[2. ]If $\chi$ has one or more coefficients that do not wrap around $\mathbb{Z}$, then we can attack search RLWE by reducing to errorless LWE and try arbitrarily many samples.
\end{itemize}

In all cases (both in LWE and RLWE), the insecurity of an instantiation is due to the fact that the error distribution is insufficiently well spread relative to the ring geometry, so, as in Section 3.2, the main lesson to learn here is that the error distribution should be taken with parameters as close as possible to those for which the the hardness theorem works (Theorem \ref{thm1}).

\subsection{Reduction and attack to PLWE}

A first fact to mention is that at the time of writing, there is no direct attack against RLWE, i.e., without a reduction to an attack on PLWE or LWE, as described in the previous subsection. So, all the attacks presented here attemp at breaking PLWE first and then to reduce RLWE to PLWE.

\begin{thm}[Elias et al. \cite{lauter}]If $K$ satisfies the following six conditions, there is a polynomial time attack to the search version of the associated RLWE scheme:
\begin{itemize}
\item[1.]$K=\mathbb{Q}(\beta)$ is Galois of degree $n$.
\item[2.]The ideal $(q)$ splits totally in $\mathcal{O}_K$.
\item[3.]$K$ is monogenic, i.e, $\mathcal{O}_K=\mathbb{Z}[\beta]$.
\item[4.] The transformation between the canonical embedding of $K$ and the power basis representation of $K$ is given by a scaled orthogonal matrix.
\item[5.]If $f$ is the minimal polynomial of $\beta$, then $f(1)\equiv 0 \pmod{q}$ .
\item[6.]The prime $q$ can be chosen suitably large. 
\end{itemize}
\end{thm}
The first two conditions are sufficient for the RLWE search-to-decision reduction in the case where $q \nmid \left[ \mathcal{O}_K: \mathbb{Z}[\beta]\right]$, which is implied by the third condition. The third and fourth conditions are sufficient for the RLWE-to-PLWE reduction; indeed, the fourth condition can be relaxed to require that the condition number of the matrix describing the transformation between the embeddings is at most polynomial in $n$, as we discussed in the previous section.

Finally, the last two conditions are sufficient for the attack on PLWE. Unfortunately (for the attacker's point of view), it is difficult to construct number fields satisfying all six conditions simultaneously. Next, we explain the attack on PLWE if 5 and 6 hold. 

Setting a s usual $\mathcal{O}=\mathbb{Z}[x]/(f(x))$, fix a public key $(a(x),b(x))\in\mathcal{O}/q\mathcal{O}\times \mathbb{T}_f$ and a secret key $s(x)\in \mathcal{O}/q\mathcal{O}$, i.e $b(x) = a(x)s(x) + e(x)$  with $e(x)$ sampled from the discrete Gaussian $\chi$.

For each root $\theta\in\mathbb{F}_q$ of $f(x)$, consider the projection $\pi_{\theta}:\mathcal{O}/q\mathcal{O}\to\mathbb{F}_q$ given by $p(x)\mapsto p(\theta)$. By \emph{short vector} in $\mathcal{O}/q\mathcal{O}$ we refer to those with \emph{small} coefficients, which in practice means that these are upper bounded, in absolute value, by $q/4$. For suitable parameter, these \emph{short vectors} lie inside a prescribed region with non-negligible probability and are easy to recognise. However, for a pair $(a(x), b(x))$, it is difficult to check if it exists $r(x)$ and a short vector $e(x)$ such that $b(x) = a(x)r(x) + e(x)$, in which case the attacker would guess that $s(x)=r(x)$. The reason is that there are $q^n$ possibilities for $s(x)$ to test, which is prohibitive. 

By contrast, in a small ring like $\mathbb{F}_q$, it is easy to examine the possibilities for $s(\theta)$ exhaustively:  we can loop through the possibilities for $s(\theta)$, obtaining for each guess $s_{\theta}$, the putative value $e(\theta) = b(\theta)-a(\theta)s_{\theta}$. The Decision Problem for PLWE, then, is solved as soon as we can recognize the set of $e(\theta)$ that arise from the Gaussian with high probability.

Again, this is difficult in general, but if 5 holds, i.e., if $\theta=1$ is a root of $f(x)$, the attacker has a chance:

Let us denote by $\mathcal{S}\subseteq \mathcal{O}/q\mathcal{O}$ the subset of polynomials that are produced by the Gaussian with non-negligible probability. This is a small set, due to the parameter choice. However, $\mathbb{F}_q$ is also a much smaller set than $\mathcal{O}/q\mathcal{O}$ and one expects that generically, $\pi_{\theta}(\mathcal{S})=\mathbb{F}_q$ or something very close. One says that in this case $\mathcal{S}$ \emph{smears} across all of $\mathbb{F}_q$. 

But we are supposing that $\theta= 1$. The polynomials $g(x)\in\mathcal{S}$ have small coefficients, and hence have small images $g(1)\in\mathbb{F}_q$. This is simply because $n$ is much smaller than $q$, due to 6, so that the sum of $n$ small coefficients is still small modulo $q$. These ideas can be turned into the following algorithm:

\begin{alg} Suppose $f(1)\equiv 0 \pmod{q}$. The input is a collection of pairs $\left\{(a_i(x),b_i(x))\in \mathcal{O}/q\mathcal{O}\times \mathbb{T}_f\right\}_{i=1}^m$, where each sample is drawn either uniformly at random or from the PLWE distribution . The output is to decide, for each sample, from which distribution is taken, with non-negligible probability. The algorithm is as follows: 

\begin{itemize}
\item[1 ] For $i=1$ to $m$ do 
\item[ ] Set $S = \mathbb{F}_q$. This is the first guess for $\pi_1(\mathcal{S})$, which will be updated after each iteration.

\begin{itemize}
\item[2] For each $s \in S$  do
\item[2.1] Compute $e_i := b_i(1)-s a_i(1)$;
\item[2.2] If $e_i$ is not small in absolute value modulo q, then conclude that the sample cannot be valid for $s$ with nonnegligible probability, and update $S=S\setminus\left\{s\right\}$;
\item[ ] Next $s$;
\end{itemize}
\item[3] If $S = \emptyset$, conclude that the sample was random, otherwise declare the sample as valid;
\item[ ] Next i;
\end{itemize}
\end{alg}
\begin{rem} Notice that in the inner loop, if the sample is valid, then $e_i=e_i(1)=\displaystyle\sum_{j=1}^ne_{ij}$, and if $\sigma$ is the variance of $\chi$ (which is spherical with respect to our embedding, fixed beforehand), then, $e_i$ is sampled from a discrete Gaussian distribution of zero mean and parameter $\sqrt{n}\sigma$. The region of non-negligible probability for this Gaussian, can be taken to be
$$
\Lambda:=\{s\in\mathbb{F}_q: |s|<n\sigma^2\leq q/4\}.
$$
\label{remark1}
\end{rem}Notice that the cyclotomic cases are protected against this attack: $\theta= 1$ is never a root modulo $q$ of a cyclotomic polynomial of degree greater than 1 when $q$ is sufficiently large. However, with minor modifications, it is possible to extend the former attack to the case where $\theta$ has small order modulo $q$.  Indeed, denote by $r$ the order of $\theta$ modulo $q$. For an unknown polynomial $e(x)$, to decide from a known value $e(\theta)$ if $e(x)$ is sampled from a Gaussian distribution in a similar fashion as in Remark \ref{remark1} is more complicated. However, one can still take advantage of a small $r$, as we explain next.

For $e(x)=\displaystyle\sum_{i=0}^ne_ix^i$, set $n=rM+l$ with $0\leq l\leq r-1$. Define $e_{Mr+k}=0$ for $0\leq k\neq l\leq r-1$ and write

$$
e(\theta)=\sum_{i=0}^{r-1}\sum_{j=0}^Me_{jr+i}\theta^i.
$$
If $e(x)$ is sampled from a multivariate Gaussian with variance very close to $\sigma^2$, then each term $\sum_{j=0}^Me_{jr+i}$ is sampled from a 1-dimensional Gaussian of variance very close to $(M+1)\sigma^2$. This defines a \emph {smallness} region , which can be pre-stored as a look-up table
$$\Lambda=\{\rho=\sum_{i=0}^{r-1}\sum_{j=0}^M\rho_{jr+i}\theta^i\subseteq\mathbb{F}_q: |\rho_{jr+i}|\leq (M+1)\sigma^2\ll q/4\}$$ 
to look at, in order to guess tentative values of $e(\theta)$. With this observation, we can derive the following algorithm:

\begin{alg} Suppose $f(\theta)\equiv 0 \pmod{q}$. The input is a collection of pairs $\left\{(a_i(x),b_i(x))\in \mathcal{O}/q\mathcal{O}\times \mathbb{T}_f\right\}_{i=1}^m$, where each sample is drawn either uniformly at random or from the PLWE distribution . The output is to decide, for each sample, from which distribution is taken, with non-negligible probability. The algorithm is as follows: 

\begin{itemize}
\item[1 ] For $i=1$ to $m$ do 
\item[ ] Set $S = \mathbb{F}_q$; 
\begin{itemize}
\item[2] For each $s \in S$  do
\item[2.1] Compute $e_i := b_i(\theta)-s a_i(\theta)$;
\item[2.2] If $e_i\not\in\Lambda$, then conclude that the sample cannot be valid for $s$ with nonnegligible probability, and update $S=S\setminus\left\{s\right\}$;
\item[ ] Next $s$;
\end{itemize}
\item[3] If $S = \emptyset$, conclude that the sample was random, otherwise declare the sample as valid;
\item[ ] Next i;
\end{itemize}
\end{alg}

\begin{rem}The third attack described in \cite{lauter} is based on the size of the residue of $e_i(\theta)$ modulo $q$. Although here the errors may take on all values in $\mathbb{F}_q$, it may still be possible to notice if the distribution of samples is not uniform. The attacking algorithm is built on a delicate probability bound in the case that $\theta\neq \pm 1$ and the order of $\theta$ modulo $q$ is not small. For example, this third attack is successful for any irreducible polynomial of degree $n=2^6$, with $q$ of the order of $2^{50}$, $\sigma=8$ and $\theta=2$.
\end{rem}

\subsection{Some number theoretical open questions motivated by attacks on PLWE}

As seen before, being simultaneously Galois and monogenic, having $\theta=1$ as a root of the minimal polynomial modulo $q$ (or some other root of small order) and the non-smearing under the evaluation map $\pi_{\alpha}$ of the set of \emph{small} vectors in $\mathcal{O}/q\mathcal{O}$ can be regarded as weakness conditions to build a RLWE-based cryptosystem. We give next a list of number theoretical problems which are motivated by the search of security in RLWE-based primitives and are still open, up to date.

\begin{question}Are there any fields of cryptographic size (i.e. $n\geq 2^{10}$) which are Galois and monogenic, other than the cyclotomic number fields and their maximal real subfields? How can one construct such fields explicitly? Is it possible to test algorithmically both features?
\end{question}
Notice that for fields of cryptographic size, the discriminant is too big to test whether or not it is square free, hence to decide if it is monogenic. An algorithmic approach which circumvects this testing is not available at the time of writing. Although for fields of small degree, a complete characterisation may be feasible (sufficient and necessary conditions for a cubic number field have been found by Gras and Archinard), the situation is much different for large degree fields. For instance, cyclic extensions tend to be non-monogenic:

\begin{thm}Any cyclic extension $K$ of prime degree $n\geq 5$ is non-monogenic except for the maximal real subfield of the $(2l+1)$-th cyclotomic field. 
\end{thm}

Another result in this direction is as follows:

\begin{thm}Let $n\geq $5 be relatively prime to $2,3$. There are only finitely many abelian number fields of degree $n$ that are monogenic. 
\end{thm}

\begin{question}Let $\theta$ be a root of  $f(x)$ modulo $q$. For which subsets $\mathcal{S}\subseteq R_q$ it is $\pi_{\theta}(\mathcal{S})=\mathbb{F}_q$? Or, at least, can one determine the conditions for non-smearing, like in the case when $\theta=1$ and $\mathcal{S}$ is a set of \emph{small} vectors in $\mathcal{O}/q\mathcal{O}$?
\end{question}
Finally, as seen before, polynomials with roots of small order modulo $q$ should be avoided. Again, cyclotomic polynomials are safe for attacks built on small order roots, as their roots have maximal order. The problem here is as follows:

\begin{question}For random polynomials $f(x)$ and random primes $q$ for which $f(x)$ has a root $\alpha$ modulo $q$, what can one say about the order of $\alpha$ modulo $q$?
\end{question}
A special instance of this question is this well-known open problem:
\begin{conj}[Artin]  Each $a\in\mathbb{Z}$ is a primitive root modulo infinitely many primes $q$ such that $a$ is not a perfect square or $-1$ modulo $4$. In fact the set of primes for which $a$ is a primitive root has density
$$
\prod_{p\mbox{ prime}}\left(1-\frac{1}{p(p-1)}\right).
$$
\end{conj}

\section{Ring Learning With Errors signatures, homomorphic encryption and some NIST figures}

\subsection{RLWE Digital Signatures}We present here a 2012 scheme by Gunyesu, Lyubashevsky and Poppelman (GLP \cite{glp}).  It has some advantages over more recent eﬃcient post-quantum digital signature proposals such as BLISS and Ring-TESLA, but although not broken, GLP as originally proposed is no longer considered to oﬀer strong levels of security. Building on GLP, A. Chopra presented GLYPH in 2017 another RLWE digital signature schemes: a special instantiation of GLP together with certain modification in the compressing and hash functions,. It is described in \cite{glph}, where a throughout analysis on its resistance to signature forgery, key-recovery, exhaustive and meet-in-the middle attacks is carried out. However, the main ideas on how to use RLWE for secure signature is already contained in GLP, hence as a first contact with the topic we have chosen this scheme.

We use the same terminology and notions as in Definition 2.3 and subsequent discussion, to which we refer the reader. This scheme uses PLWE in the cyclotomic ring $R_q=\mathbb{F}_q[x]/(\Phi_n(x))$ with $q$ an odd prime congruent to $1$ mod $4$ or a power of $2$.

A first difference to mention here is that instead of discrete Gaussians, the coefficients of \emph{small} polynomials are sampled uniformly from $\left\{-1,0,1\right\}$ modulo $q$. This version of RLWE, is called the Compact Knapsack Problem over ideal lattices, whose decisional version backs GLP. Secondly, the lengths of signatures must not exceed a prescribed parameter $n$, regardless of the size of the message to sign. To attain this, the scheme uses a) a hash function $H$\footnote{A hash function is $H:\cup_{r\geq 1}\mathbb{F}_2^r\to\mathbb{F}_2^{\kappa}$ with fixed $\kappa$. In GLP/GLYPH, a common choice for $H$ is the function SHA256.}, which accepts bit strings of arbitrary length and returns bit strings of bounded length, and b) a function $F$ from the target of $H$ to the set of polynomials of degree $m=\phi(n)$ with exactly $k$ of their coefficients having absolute in $\pm 1$ and the rest being zero such that the probability of mapping two hash outputs to the same sparse element is less than $1/2^{\lambda}$, where $\lambda$ is a security parameter. 

Hence, the procedure has a sampling rejection step, which ensures that the output signature is not exploitably correlated with the signer's secret key values: if the infinity norm of a signature polynomial exceeds a fixed bound, $\beta$, that polynomial will be discarded and the signing process starts again. This process will be repeated until the infinity norm of the signature polynomial is less than or equal $\beta=k-1$, where $k$ is the number of non-zero coefficients allowed in acceptable polynomials. 

Third, it is necessary to fix an injective map $I:R_q\to \mathbb{F}_2^N$, with $N\gg 1$. Last, the maximum degree of the signature polynomials will be $m-1$ so that there are $m$ coefficients. Typical values for $m$ are 512, and 1024. For $m=1024$, GLYPH sets $q = 59393$, $b=16383$ and $k=16$. The scheme is as follows:

\begin{itemize}
\item[1. ]Key generation:
	\begin{itemize}
		\item[1.1] Generate, uniformly, two small polynomials $s(x)$ and $e(x)$. The pair $(s(x),e(x))$ is the private key.
		\item[1.2] Compute $t(x) = a(x)s(x) + e(x)$, with $a(x)$ chosen uniformly at random. The public key is $(a(x),t(x))$.
	\end{itemize}
\item[2. ]Signature generation:
	\begin{itemize}
		\item[2.1] Input: a message $m(x)\in R_q$ and $(a(x),e(x),s(x))$
		\item[2.2] Generate two small polynomials $y_1(x)$ and $y_2(x)$.
		\item[2.3] Compute $w(x) = a(x)y_1(x) + y_2(x)$.
		\item[2.4] Set $\omega=I(w(x))$ and $\mu=I(m(x))$.
		\item[2.5] Compute $c(x) = F(H(\omega | |\mu))$. The symbol $||$ denotes concatenation of strings.
		\item[2.6] Compute $z_1(x) = s(x)c(x) + y_1(x)$ and $z_2(x) = e(x)c(x) + y_2(x)$.
		\item[2.7] While the infinity norms of $z_1(x)$ or $z_2(x)$ is  greater than $\beta$ go to step 2.1. 
		\item[2.8] Output: $(c(x), z_1(x),z_2(x))$. Transmit the signature along with the message $m(x)$. Notice that we are not discussing here signatures of encrypted messages, which is a more sophisticated cryptographic functionality.
	\end{itemize}
\item[3. ]Signature verification:
	\begin{itemize}
		\item[3.1] Input: $(c(x), z_1(x),z_2(x),m(x))$.
		\item[3.2] Verify that the infinity norms of $z_1(x)$ and $z_2(x)$ do not exceed $\beta$. If not, reject the signature.
		\item[3.3] Compute $w'(x) = a(x)z_1(x) + z_2(x) - t(x)c(x)$.
		\item[3.4] Set $\omega'=I(w'(x))$ and $\mu=I(m(x))$.
		\item[3.5] Compute $c'(x) = F(H(\omega'|| \mu))$.
		\item[3.6] Output: If $c'(x) \neq c(x)$ reject the signature, otherwise accept the signature as valid.
	\end{itemize}
\end{itemize}
Notice that $a(x)z_1(x) + z_2(x) - t(x)c(x)=w(x)$, hence $c'(x)=c(x)$ if the signature is not tampered, hence the scheme is correct.

\begin{rem}The private key $(s(x),e(x))$ can be represented in $2n\log_2(3)$ bits of memory, and the public key $a(x)s(x) + e(x)$ can be represented in $n\log_2(q)$ bits, which makes GLP feasible for practical implementations.
\end{rem}

\begin{rem}Both in \cite{glph} and in the earlier \cite{glp}, the application of the hash function $H$ may result unclear for a non experienced reader. The reason is that in what we have labeled steps 2.5 and 3.4, both schemes apply $H$, defined over a binary domain, to inputs which are not binary. This point is probably not taken very seriously by the experts, for all what matters is that $H$ is a collision resistant function and, more important, that when it comes to comparing $H(\omega | |\mu)$ with $H(\omega' | |\mu)$, they can only be equal with overwhelming probability if and only if $\omega=\omega'$. But of course one needs to make binary the arguments $w(x)$ and $m(x)$ of $H$, and this is why we have fixed the innaccuracy by resourcing to a function $I$ which injectively outputs binary strings on polynomial inputs and defined $\omega=I(w(x))$ and $\mu=I(m(x))$. In \cite{glp} page 6 it is discussed how forging a signature implies finding a collision on $H$.
\end{rem}
\subsection{RLWE Homomorphic encryption}Homomorphic encryption was first introduced by Rivest, Adleman and Dertouzos back in the 70's (\cite{he70s}), where they raised the problem of constructing a fully homomorphic scheme (a \emph{privacy homomorphism}, using their phraseology). This problem was solved by Craig Gentry in 2009 in its seminal paper \cite{gentry}, by using ideal lattices and (essentially) a modified version of PLWE. The possibility of cheap cloud computing and distributed storage has drastically changed how business and individuals process their data and although traditional encryption like AES are very fast, to perform even simple analytics on encrypted data requires either the cloud server to access the secret keys, leading to security concerns or to download the data, decrypt and operate, which is costly. Homomorphic encryption is the solution to this challenge.

Areas where homomorphic encryption has applications include e-voting systems (\cite{gama}) and processing or computing on encrypted health, financial or other kinds of sensitive data on external servers like cloud or distributed devices.

Homomorphic and fully homomorphic encryption(FHE)  has already been introduced here in Definition 2.5, and Example 2.6 provides an example of a homomorphic but not non-fully homomorphic encryption scheme. 

\begin{examples}Another example of homomorphic encryption is the LWE cryptosystem. To avoid entering into technicalities, choose an odd prime $q$, so that $2$ is invertible in $\mathbb{F}_q$. We observe that a LWE-oracle is \emph{essentially} homomorphic:
given a private key $s\in\mathbb{F}_q^n$, two uniformly sampled vectors $a_1,a_2\in\mathbb{F}_q^n$ and two errors $e_1,e_2$ taken from a $\mathbb{T}$-valued random variable $\chi$, of $0$-mean and variance $\sigma^2$, we se that
$$
(a_1,\langle a_1,s \rangle+e_1)+(a_2,\langle a_2,s \rangle+e_2)=(a_1+a_2,\langle (a_1+a_2),s \rangle+(e_1+e_2)).
$$
\emph{Essentially} means that the sum $e_1+e_2$ is taken from the variable $2\chi$, which has also $0$-mean but variance $2\sigma^2$. This easy observation allows to define a homomorphic cryptosystem, which is a minor modification of Regev's scheme presented in Section 3. However, if we keep adding encryption of data, this results in amplifying the error of the final encrypted data, and when this error passes a certain threshold, decryption becomes impossible. This implies that the length of the arithmetic circuit must be known beforehand and the parameters must be set to meet this feature.
\end{examples}

An analogous analysis as in the previous example shows that RLWE oracles are also essentially homomorphic both in the additive and multiplicative structure, where, again, essentially means that the error of the sum/product is an amplification of the individual errors of the encrypted data, hence, RLWE provides a FHE scheme, as we see next.

\begin{defn}[The BGV cryptosystem (\cite{bgv}, Section 3.4)] 
\hfill \break
Denote $R=\mathbb{Z}[x]/(\Phi_n(x))$, with $\Phi_n(x)$ the $n$-th cyclotomic polynomial and set $R_N:=R/NR$. Consider as the space of plaintexts the ring $R_{p^r}$, for fixed $r$ and prime $p$. The scheme is parametrized by a sequence of decreasing moduli $q_L>q_{L-1}>...>q_0$ such that $q_i\leq\min\left\{\sqrt{q_{i+1}},\frac{q_{i+1}}{2}\right\}$ and an $i$-th level ciphertext is a vector $(v,w)\in R_{q_i}^2$.
\begin{itemize}
\item[1. ]Key generation: Chose $s\in R$ by sampling from a discrete Gaussian such that the probability of the set $\left\{0,\pm 1\right\}^{\phi(n)}$ is close enough to $1$.
\item[2. ]Encryption/Decryption: A plaintext $\alpha\in R_{p^r}$ is encrypted to $E(s,\alpha)=(p_0,p_1)\in R_{q_i}^2$ if and only if $p_0+sp_1$ modulo $q_i$ equals $\alpha+p^r\epsilon$ in $R$ with $||\epsilon||<q_i/p^r$ for some $i\in\{0,...,L\}$.
\end{itemize}
\end{defn}

Observe that adding or multiplying two i-level ciphertexts results in an $i+1$-level ciphertext, so computations over level $L$-ciphertexts are not allowed, as they cannot be decrypted. Several recent refinements to this scheme have been proposed (\cite{halsh}) and the topic is still under research.

A number of open-source implementations of homomorphic encryption are available. For instance, HELib, a widely used library from IBM that implements the BGV cryptosystem, SEAL, a Microsoft version, $\Lambda O \lambda$ (pronounced \emph{LOL}), a Haskell library for ring-based lattice cryptography that supports FHE or PALISADE, a general lattice encryption library. it is possible to add new implementations after public review by contacting contact@HomomorphicEncryption.org. In sum, homomorphic encryption is already ripe for mainstream use but the lack of standardisation makes difficult to decide on which implementation to use.

\subsection{NIST figures} In 2017, the American National Institute of Standards and Technology (NIST), launched an open call (https://csrc.nist.gov/Projects/Post-Quantum-Cryptography) to evaluate and standardize one or more quantum-resistant public-key cryptographic algorithms. In their own words: 

\emph{The question of when a large-scale quantum computer will be built is a complicated one. While in the past it was less clear that large quantum computers are a physical possibility, many scientists now believe it to be merely a significant engineering challenge. Some engineers even predict that within the next twenty or so years sufficiently large quantum computers will be built to break essentially all public key schemes currently in use. Historically, it has taken almost two decades to deploy our modern public key cryptography infrastructure.  Therefore, regardless of whether we can estimate the exact time of the arrival of the quantum computing era, we must begin now to prepare our information security systems to be able to resist quantum computing.}

The deadline for submission was November 30, 2017. The total number of submissions (for encryption, key exchange and signatures) was 71. In the first round, 14 submissions were attacked or withdrawn. Of the remaining 57, some of the proposals (mainly code-based ones) did merge. Taking this into account, 50 proposals remained unbroken. Some of them were found to have non-fatal attacks, which can be avoided with a right choice of parameters, also in the first round. 

Of these 50 proposals: 9 were code-based, 21 lattice-based, 2 hash-based , 9 multivariate-based, 1 supersingular isogeny Diffie-Hellman (SIDH) key-exchange protocol. The remaining 8 submissions were hybrid or based on problems such as random walks (1), braids (2), Chebychev polynomials (1) or hypercomplex numbers (1).

In January 2019, a second round started and taking into account the attacks and feedback to the surviving proposals of the first round, 26 proposals have passed this new sieve. The numbers of remaining proposals (at the time of writing) within each category are listed in the following table, constructed out of data from https://www.safecrypto.eu/pqclounge/:

\begin{table}[htbp]
\begin{center}
\begin{tabular}{|l|l|}
\hline
Category & Number of proposals\\
\hline \hline
Code-based (Hamming) &  5 \\ \hline
Code-based (rank metric) &  2 \\ \hline
Lattice-based (LWE)  &  1 \\ \hline
Lattice-based (RLWE)  &  6 \\ \hline
Lattice-based (PLWE)  &  1 \\ \hline
Lattice-based (Other)  &  4 \\ \hline
Multivariate-based & 4 \\ \hline
Hash-based &  1 \\ \hline
Supersingular isogeny-based &  1 \\ \hline
Other &  1 \\ \hline
\end{tabular}
\caption{NIST proposals. Second Round.}
\label{tabla:sencilla}
\end{center}
\end{table}

%Again, some of the proposals in the first round have merged in the second round. And again, the RLWE category remains the strongest contender in terms of number of surviving proposals. At  https://www.safecrypto.eu/pqclounge/ a summary of candidates and the history of all submissions, attacks and withdrawals is available to filter and check.


\begin{thebibliography}{99}
\bibitem{ADPS} E. Alkim, L. Ducas, T. P\"oppelmann, P. Schwabe. Post-quantum key exchange: a new hope. Proceedings of the 25th USENIX Security Symposium 2016 pp 327--343
\bibitem{alsu} M. H. Alsuwaiyel. \emph{Algorithms: Design Techniques and Analysis.} World Scientific, 1999. 
\bibitem{ag11} S. Arora, R. Ge. New algorithms for learning in presence of errors. In \emph{Automata, Languages and Programming}, pages 403–415. Springer, 2011. 
\bibitem{bang}A.S. Bang: Om ligningen $\Phi_m(X)=0$. \emph{Nyt tidsskrift for Matematik, Afdeling B} (1895), 6--12.
\bibitem{bateman2}P.T. Bateman: On the size of the coefficients of the cyclotomic polynomial.\emph{Seminaire de Th\'eorie des Nombres de Bordeaux, 11} (28) (1982) 1--18. 
\bibitem{blanco} I. Blanco-Chac\'on. On the RLWE/PLWE equivalence for cyclotomic number fields. To appear in \emph{Applicable Algebra in Engineering, Communications and Computing}, 2020 (available in arxiv: https://arxiv.org/abs/2001.10891 )
\bibitem{cryptostanford} D. Boneh, V. Shoup. A graduate course in applied cryptography, 2020 https://crypto.stanford.edu/\~~dabo/cryptobook/BonehShoup\_0\_5.pdf
\bibitem{bernstein} D.J. Bernstein, C. Chuengsatiansup, T. Lange, C. van Vredendaal: NTRU Prime (2016). http://eprint.iacr.org/2016/461
\bibitem{boas} P. E. Boas. Another NP-Complete Problem and the Complexity of Computing Short Vectors in a Lattice. Tech. Report 81-04, Mathematische Instituut, University of Amsterdam, 1981.
\bibitem{google} M. Braithwaite.  Experimenting with post-quantum cryptography.  Google Security Blog, 2016. https://security.googleblog.com/2016/07/experimenting-with-post-quantum.html
\bibitem{bgv}Z. Brakersky, C. Gentry, V. Vaikuntanathan. (Leveled) Fully Homomorphic Encryption without Bootstrapping. https://people.csail.mit.edu/vinodv/6892-Fall2013/BGV.pdf
\bibitem{gama}I. Chillotti, N. Gama, M. Georgieva and M. Izabach\'ene. An homomorphic LWE based e-voting scheme. $https://ilachill.github.io/papers/CGGI16a-An_homomorphic_LWE_based_E-voting_Scheme.pdf$
\bibitem{glph}A. Chopra. GLYPH: A New Instantiation of the GLP Digital Signature Scheme. https://eprint.iacr.org/2017/766.pdf
\bibitem{ding} J. Ding, B.Y. Yang. Multivariate public key cryptography. $https://link.springer.com/content/pdf/10.1007/978-3-540-88702-7_6.pdf$
\bibitem{dingkex} J. Ding, X. Xiang, L. Xiaodong. A Simple Provably Secure Key Exchange Scheme Based on the Learning with Errors Problem (2012) https://eprint.iacr.org/2012/688.pdf
\bibitem{DD} L. Ducas, A. Durmus. Ring-LWE in polynomial rings. In PKC, 2012.
\bibitem{statistics}N.C. Dwarakanath, S. D. Galbraith. Sampling from discrete Gaussians for lattice-based cryptography on a constrained device. Preprint: https://www.math.auckland.ac.nz/~sgal018/gen-gaussians.pdf
\bibitem{lauter} Y. Elias, K. Lauter, E. Ozman, K. Stange. Ring-LWE cryptography for the number theorist. In: E. Eischen, L. Long, R. Pries, K. Stange (eds) Directions in Number Theory. Association for Women in Mathematics Series, vol 3. Springer 2016.
\bibitem{leo} L. de Feo: Mathematics of isogeny based cryptography. https://arxiv.org/pdf/1711.04062.pdf
\bibitem{gentry}C. Gentry. Fully Homomorphic Encryption Using Ideal Lattices. In the 41st ACM Symposium on Theory of Computing (STOC), 2009. 
\bibitem{glp}T. Guneysu, V. Lyubashevsky and T. Poppelmann. Practical lattice-based cryptography: A signature scheme for embedded systems. https://www.iacr.org/archive/ches2012/74280529/74280529.pdf
\bibitem{halsh}S. Halevi, V. Shoup. Faster Homomorphic Linear Transformations in HElib. https://eprint.iacr.org/2018/244.pdf
\bibitem{speed}D. Harvey. Faster arithmetic for number-theoretic transforms. \emph{J. Symb. Comput.}, 60:113–119, 2014.
\bibitem{ntru} J. Hoffstein, J. Pipher, J. H. Silverman.  NTRU: A ring-based public key cryptosystem. ANTS-III, pp 267-–288, 1998.
\bibitem{nfs} A. Joux A, R. Lercier. Number Field Sieve for the DLP. In: van H.C.A. Tilborg, S. Jajodia (eds) \emph{Encyclopedia of Cryptography and Security}. Springer, Boston (2011).
\bibitem{kobliz} N. Kobliz. \emph{P-adic numbers, p-adic analysis, and zeta-functions}. Graduate texts in Mathematics, n. 58. Springer 1984.
\bibitem{leclerq} R.Le Clercq, S. S. Roy, F. Vercauteren, I. Verbauwhede. Efficient software implementation of RLWE encryption. https://eprint.iacr.org/2014/725.pdf
\bibitem{LPR} V. Lyubashevsky, C. Peikert, O. Regev. On ideal lattices and learning with errors over rings. In: Gilbert H. (eds) \emph{Advances in Cryptology – EUROCRYPT 2010.} Lecture Notes in Computer Science, 6110. Springer.
\bibitem{menezes} A. Menezes, T. Okamoto and S. Vanstone. Reducing elliptic curve logarithms to logarithms in a finite field. \emph{IEEE Transactions on Information Theory} 39(5) 1639--1646, 1993.
\bibitem{micciancio}D. Micciancio.The shortest vector in a lattice is hard to approximate to within some constant.In Proc.39th Annual IEEE Symposium on Foundations of Computer Science, 1998.
\bibitem{oberbeck} R. Overbeck, N. Sendrier. Code-based cryptography. In: D.J. Bernstein, J. Buchmann, E. Dahmen (eds) Post-Quantum Cryptography  (2009). Springer, Berlin, Heidelberg .
\bibitem{peikert} C. Peikert. How (not) to instantiate ring-RLWE. In Zikas, V.; de Prisco, R. (eds.) SCN 2016, LNCS vol 9841, pags. 411--430 (2016) Springer.
\bibitem{peikertkex}C. Peikert. Lattice Cryptography for the Internet  (2014): https://eprint.iacr.org/2014/070.pdf
\bibitem{PRS} C. Peikert, O. Regev, N. Stephens-Davidowitz. Pseudorandomness of Ring-LWE for any ring and modulus. In STOC, 2017.
\bibitem{regev} O. Regev. On lattices, learning with errors, random linear codes and cryptography. J. ACM, 56 (6), 2009.
\bibitem{he70s} R. Rivest, L. Adleman, and M. Dertouzos. On data banks and privacy homomorphisms. In \emph{Foundations of Secure Computation}, pp. 169–180, 1978.
\bibitem{RSW} M. Rosca, D. Stehl\'e, A. Wallet. On the ring-LWE and polynomial-LWE problems. In: Nielsen J., Rijmen V. (eds) \emph{Advances in Cryptology – EUROCRYPT 2018.} Lecture Notes in Computer Science, vol 10820. Springer.
\bibitem{scott} M. Scott. A note on the implementation of the Number Field Transform. IMACC 2017. https://eprint.iacr.org/2017/727.pdf
\bibitem{shor} P. Shor. Polynomial-Time Algorithms for Prime Factorization and Discrete Logarithms on a Quantum Computer. \emph{SIAM Journal on Computing, 26}  (5) (1997) 1484--1509.
\bibitem{stehle2}D. N. Stehle, R. Steinfeld, K. Tanaka, K. Xagawa. Efficient public key encryption based on ideal lattices. In \emph{Advances in Cryptology ASIACRYPT 2009}. 617--635 (2009).
\bibitem{stewart} I. Stewart. \emph{Algebraic number theory and Fermat's last theorem.} AK Peters Ltd, 2002.
\bibitem{was} L. Washington. \emph{Introduction to Cyclotomic Fields}, Graduate Texts in Mathematics, 83 (2 nd ed.), Springer-Verlag (1997).
\end{thebibliography}
\end{document}